\documentclass[a4paper,reqno,11pt]{article}

\usepackage[hmargin=2cm,vmargin=2cm]{geometry}

\usepackage[utf8]{inputenc}
\usepackage[T1]{fontenc}

\usepackage{amsmath,amssymb,amsthm,mathrsfs,stmaryrd,mathdots,color,graphicx,framed,authblk}

\usepackage[colorlinks=true, pdfstartview=FitV, urlcolor=blue, citecolor=red, linkcolor=blue]{hyperref}

\definecolor{shadecolor}{rgb}{0.9, 0.9, 0.86}

\def\Re{\mathrm {Re}\,}
\def\Im{\mathrm {Im}\,}
\def\wt{\widetilde}
\def\wh{\widehat}
\def\C{\mathbb{C}}
\def\R{\mathbb{R}}

\def\G{\Gamma}
\def\1{\mathbf{1}}
\def\H{\mathrm{H}}
\def\b{\beta}
\def\a{\alpha}

\def\d{\mathrm d}
\def\O{\mathcal{O}}
\def\e{\mathrm{e}}
\def\i{\mathrm{i}}
\def\t{\mathbf{t}}
\def\u{\mathbf{u}}
\def\pa{\partial}
\def\res{\mathop{\mathrm {res}}\limits_}
\def\tr{\mathrm{tr}\,}
\def\s{\sigma}

\def\0{\boldsymbol{0}}
\def\S{\mathfrak{S}}
\def\ll{\lambda}

\def\H{\mathcal{H}}
\def\U{\mathrm{U}}
\def\J{\mathcal{J}}
\def\CC{\mathcal{C}}
\def\cyc{\mathrm{cyc}}
\def\Id{{\rm Id}}
\def\P{\mathcal{P}}

\def\be{\begin{equation}}
\def\ee{\end{equation}}

\newtheorem{theorem}{Theorem}[section]
\newtheorem{definition}[theorem]{Definition}
\newtheorem{example}[theorem]{Example}
\newtheorem{lemma}[theorem]{Lemma}
\newtheorem{remark}[theorem]{Remark}
\newtheorem{proposition}[theorem]{Proposition} 
\newtheorem{corollary}[theorem]{Corollary}

\newmuskip\pFqmuskip
\newcommand*\pFq[6][8]{
  \begingroup 
  \pFqmuskip=#1mu\relax
  \mathchardef\normalcomma=\mathcode`, 
  \mathcode`\,=\string"8000            
  \begingroup\lccode`\~=`\,
  \lowercase{\endgroup\let~}\pFqcomma  
  {}_{#2}F_{#3}{\left(\left.\genfrac..{0pt}{}{#4}{#5}\right|#6\right)}
  \endgroup
}
\newcommand{\pFqcomma}{{\normalcomma}\mskip\pFqmuskip}

\begin{document}

\pagestyle{myheadings}
\numberwithin{equation}{section}

\title{Jacobi Ensemble, Hurwitz Numbers and Wilson Polynomials}
\date{}

\author{Massimo Gisonni\footnote{\text{SISSA, Trieste, Italy;} \texttt{massimo.gisonni@sissa.it}}, Tamara Grava\footnote{\text{School of Mathematics, Bristol University, UK} and \text{SISSA, Trieste, Italy;} \texttt{grava@sissa.it}}, and Giulio Ruzza\footnote{\text{IRMP, UCLouvain, Louvain-la-Neuve, Belgium;} \texttt{giulio.ruzza@uclouvain.be}}}

\maketitle

\begin{abstract}
We express the topological expansion of the Jacobi Unitary Ensemble in terms of triple monotone Hurwitz numbers.
This completes the combinatorial interpretation of the topological expansion of the classical unitary invariant matrix ensembles.
We also provide effective formul\ae\ for generating functions of multipoint correlators of the Jacobi Unitary Ensemble in terms of Wilson polynomials, generalizing the known relations between  one point correlators and Wilson polynomials.
\end{abstract}

\begin{flushright}
{\it 
To the memory of Boris Dubrovin,\\
our mentor, teacher, \\
and a source of ever-lasting inspiration. 
}
\end{flushright}

\section{Introduction and results}

Throughout this paper we denote $\H_N(I)$ the set of hermitian matrices of size $N=1,2,\dots$ with eigenvalues in the interval $I\subseteq \R$. In particular $\H_N(I)$ can be endowed with the Lebesgue measure
\be
\d X=\prod_{i=1}^N\d X_{ii}\prod_{1\leq i<j\leq N}\d \Re X_{ij}\,\d \Im X_{ij}.
\ee
The \emph{Jacobi Unitary Ensemble} (JUE) is defined by the following measure on $\H_N(0,1)$
\be
\label{JUEmeasure}
\d m_N^{\sf J}(X)=\frac 1{C_N^{\sf J}}{\det}^\a (X){\det}^\b(\1-X)\d X,
\ee
with parameters $\a,\b$ satifying $\Re\a,\Re\b>-1$.
The normalizing constant
\be \label{NormalizationC}
C_N^{\sf J}=\int_{\H_N(0,1)}{\det}^\a (X){\det}^\b(\1-X)\d X=\pi^{\frac{N(N-1)}{2}} \prod_{k=0}^{N-1} \frac{\Gamma (\a+k+1) \Gamma (\beta+k+1)}{\G(\a+\beta+2N-k)}
\ee
ensures that $\d m_N^{\sf J}$ has total mass $1$; the above integral can be computed by a standard formula \cite{Deift} in terms of the norming constants $h_\ell^{\sf J}$ of the monic Jacobi polynomials, see \eqref{NormingC}.

If $\a$ and $\b$ are integers, so that $M_\a=\a+N$ and $M_\b=\b+N$ are integers, the probability measure \eqref{JUEmeasure} describes the distribution of the matrix $X=(W_A+W_B)^{-1/2}W_A(W_A+W_B)^{-1/2}\in\H_N(0,1)$ where $W_A=A^\dagger A$ and $W_B=B^\dagger B$ are the \emph{Wishart matrices} associated with the random matrices $A,B$ of size $M_\a \times N,M_\beta \times N$ respectively, with i.i.d. Gaussian entries \cite{F2010}.

Given positive integers $k_1,\dots,k_\ell\geq 0$ we shall consider the expectation values
\be
\label{correlatorsintro}
\left\langle\prod_{j=1}^\ell\tr X^{\pm k_j}\right\rangle:=\int_{\H_N(0,1)}\left(\prod_{j=1}^\ell\tr X^{\pm k_j}\right)\d m_N^{\sf J}(X),\quad
\ee
which we term (respectively, positive and negative) \emph{JUE correlators}.

\begin{remark}
Although \eqref{correlatorsintro} is defined only for $\Re\a\pm\sum_{i=1}^\ell k_i>-1,\Re\b>-1$, it will be clear from the formul\ae\ of Corollary \ref{corollaryWilson} below that the JUE correlators extend to rational functions of $N,\a,\b$.
\end{remark}

\subsection{JUE correlators and Hurwitz numbers}

Our first result gives a combinatorial interpretation for the large $N$ \emph{topological expansion} \cite{ME2003,EMP2008,EO2007} of JUE correlators \eqref{correlatorsintro}.
This provides an analogue of the classical result of Bessis, Itzykson and Zuber \cite{BIZ} expressing correlators of the Gaussian Unitary Ensemble as a generating function counting ribbon graphs weighted according to their genus, see also \cite{ME2003}.
At the same time, it is more similar in spirit (and actually a generalization, see Remark \ref{LaguerreLimit}) of the analogous result for the Laguerre Unitary Ensemble, whose correlators are expressed in terms of \emph{double} monotone Hurwitz numbers \cite{CDO2018}, and (for a specific value of the parameter) in terms of Hodge integrals \cite{DY2017b,DLYZ2020,GGR2020}; in particular in \cite{GGR2020} we provide an ELSV-type formula \cite{ELSV2001} for weighted double monotone Hurwitz numbers in terms of Hodge integrals.

Our description of the JUE correlators involves triple monotone Hurwitz numbers, which we promptly define; to this end let us recall that a \emph{partition} is a sequence $\ll=(\ll_1,\dots,\ll_{\ell})$ of integers $\ll_1\geq\dots\geq\ll_\ell>0$, termed \emph{parts} of $\ll$; the number $\ell$ is called length of the partition, denoted in general $\ell(\ll)$, and the number $|\ll|=\sum_{j=1}^{\ell}\ll_j$ is called weight of the partition.
We shall use the notation $\ll\vdash n$ to indicate that $\ll$ is a partition of $n$, i.e. $|\ll|=n$.

We denote $\S_n$ the group of permutations of $\{1,\dots,n\}$; for any $\ll\vdash n$ let $\cyc(\ll)\subset\S_n$ the conjugacy class of permutations of cycle-type $\ll$.
It is worth recalling that the centralizer of any permutation in $\cyc(\ll)$ has order
\be
\label{zll}
z_\ll=\prod_{i\geq 1}i^{m_i}(m_i!),\qquad m_i=\left|\left\lbrace j:\ \ll_j=i\right\rbrace\right|,
\ee
where the symbol $|\cdot|$ denotes the cardinality of the set.

Hurwitz numbers were introduced by Hurwitz to count the number of non-equivalent branched coverings of the Riemann sphere with a given set of branch points and branch profile \cite{Hurwitz}.
This problem is essentially equivalent to count factorizations in the symmetric groups with permutations of assigned cycle-type and, possibly, other constraints.
It is a problem of long-standing  interest  in combinatorics, geometry, and physics \cite{O2000,GGPN2014,HO2015,GPH2015,GGPN2016,GGPN2017,ACEH2020}.  

The type of Hurwitz numbers relevant to our study is defined as follows.

\begin{shaded}
\begin{definition}\label{def:hurwitz}
Given $n\geq 0$, three partitions $\ll,\mu,\nu\vdash n$ and an integer $g\geq 0$, we define $h_g(\ll,\mu,\nu)$ to be the number of tuples $(\pi_1,\pi_2,\tau_1,\dots,\tau_r)$ of permutations in $\S_n$ such that
\begin{enumerate}
\item $r=2g-2-n+\ell(\mu)+\ell(\nu)+\ell(\ll)$,
\item $\pi_1\in\cyc(\mu)$, $\pi_2\in\cyc(\nu)$, 
\item $\tau_i=(a_i,b_i)$ are transpositions, with $a_i<b_i$ and $b_1\leq\cdots\leq b_r$, and
\item $\pi_1\pi_2\tau_1\cdots\tau_r\in\cyc(\ll)$.
\end{enumerate}
\end{definition}
\end{shaded}

The relation of these Hurwitz numbers to the JUE is expressed by the following result.
\begin{shaded}
\begin{theorem}\label{thmhurwitz}
Under the re-scaling $\a=(c_\a-1)N$, $\b=(c_\b-1)N$, for any partition $\ll$ we have the following Laurent expansions as $N\to\infty$;
\begin{align}
(-1)^{|\ll|}N^{\ell(\ll)}\frac {|\ll|!}{z_\ll}\left\langle\prod_{j=1}^\ell\tr X^{\ll_j}\right\rangle&=
\sum_{g\geq 0}\frac 1{N^{2g-2}}\sum_{\mu,\nu\vdash|\ll|}
\frac{c_\a^{\ell(\nu)}}{(-c_\a-c_\b)^{\ell(\mu)+\ell(\nu)+\ell(\ll)+2g-2}}
h_g(\ll,\mu,\nu),
\\
(-1)^{|\ll|}N^{\ell(\ll)}\frac {|\ll|!}{z_\ll}\left\langle\prod_{j=1}^\ell\tr X^{-\ll_j}\right\rangle&=
\sum_{g\geq 0}\frac 1{N^{2g-2}}\sum_{\mu,\nu\vdash|\ll|}
\frac{\left(1-c_\a-c_\b\right)^{\ell(\nu)}}{(c_\a-1)^{\ell(\mu)+\ell(\nu)+\ell(\ll)+2g-2}}
h_g(\ll,\mu,\nu),
\end{align}
where $z_\ll$ is given in \eqref{zll} and $h_g(\ll,\mu,\nu)$ are the monotone triple Hurwitz numbers of Definition \ref{def:hurwitz}.
\end{theorem}
\end{shaded}

The proof is in Section \ref{sec2}. There is a similar result for the Laguerre Unitary Ensemble (LUE) \cite{CDO2018} which is recovered by the limit explained in Remark \ref{LaguerreLimit}. However, the proof presented in this paper uses substantially different methods than those employed in \cite{CDO2018}; in particular our proof is completely self-contained and uses the notion of \emph{multiparametric weighted Hurwitz numbers}, see e.g. \cite{HO2015} and Section \ref{secHurwitz}.

\begin{remark}[Connected correlators and connected Hurwitz numbers] \label{ConnectedH}
By standard combinatorial methods \cite{Stanley} it is possible to conclude from Theorem \ref{thmhurwitz} that the \emph{connected} JUE correlators
\be
\label{connectedcorrelators}
\left\langle\prod_{j=1}^\ell\tr X^{\pm\ll_j}\right\rangle^{\sf c}=\sum_{\mathcal{P}\text{ partition of }\{1,\dots,\ell\}}(-1)^{|\mathcal{P}|-1}(|\mathcal{P}|-1)!\prod_{A\in\mathcal{P}}\left\langle\prod_{a\in A}\tr X^{\pm\ll_a}\right\rangle
\ee
admit the same large $N$ expansions as in Theorem \ref{thmhurwitz}, with the Hurwitz numbers $h_g(\ll,\mu,\nu)$ replaced by their \emph{connected} counterparts $h_g^{\sf c}(\ll,\mu,\nu)$.
The latter are defined as the number of tuples $(\pi_1,\pi_2,\tau_1,\dots,\tau_r)$ satisfying {\rm (1)--(4)} in Definition \ref{def:hurwitz} and the additional constraint that the subgroup generated by $\pi_1,\pi_2,\tau_1,\dots,\tau_r$ acts transitively on $\{1,\dots,n\}$.
\end{remark}

\subsection{Computing correlators of hermitian models}

To provide an effective computation of the JUE correlators we first consider the general case of a measure on $\H_N(I)$ of the form
\be
\label{generalmeasure}
\d m_N(X)=\frac 1{C_N}\exp\tr V(X)\d X,
\ee
with normalizing constant $C_N=\int_{\H_N(I)}\exp\tr V(X)\d X$. 
Here $V(x)$ is a smooth function of $x\in I^\circ$ (the interior of $I$) and we assume that $\exp V(x)=\O\left(|x-x_0|^{-1+\varepsilon}\right)$ for some $\varepsilon>0$ as $x\in I^\circ$ approaches a finite endpoint $x_0$ of $I$; further, if $I$ extends to $\pm\infty$ we assume that $V(x)\to-\infty$ fast enough as $x\to\pm\infty$ in order for the measure \eqref{generalmeasure} to have finite moments of all orders, so that the associated orthogonal polynomials exist.
The expression $\tr V(X)$ in \eqref{generalmeasure} for an hermitian matrix $X$ is defined via the spectral theorem. 
The JUE is recovered for $I=[0,1]$ and $V(x)=\a\log x+\b\log(1-x)$, $\Re \a,\Re\b>-1$.

Introduce the \emph{cumulant functions}
\be
\label{cumulantfunctions}
\mathscr C_{\ell}(z_1,\dots,z_\ell):=\int_{\H_N(I)}\prod_{i=1}^\ell\tr\left[\left(z_i-X\right)^{-1}\right]\d m_N(X),\qquad\ell\geq 1,
\ee
which are analytic functions of $z_1,\dots,z_\ell\in\C\setminus I$, symmetric in the variables $z_1,\dots,z_\ell$.
To simplify the analysis it is convenient to introduce the \emph{connected} cumulant functions
\be
\label{disctoconn}
\mathscr C_{\ell}^{\sf c}(z_1,\dots,z_\ell)=
\sum_{\mathcal{P}\text{ partition of }\{1,\dots,\ell\}}(-1)^{|\mathcal{P}|-1}(|\mathcal{P}|-1)!
\prod_{A\in\mathcal{P}}\mathscr C_{|A|}(\{z_a\}_{a\in A}),
\ee
from which the cumulant functions can be recovered by
\be
\label{conntodisc}
\mathscr C_{\ell}(z_1,\dots,z_\ell)=\sum_{\mathcal{P}\text{ partition of }\{1,\dots,\ell\}}\prod_{A\in\mathcal{P}}\mathscr C_{|A|}^{\sf c}(\{z_a\}_{a\in A}).
\ee
For example, $\mathscr C_1(z)=\mathscr C_1^{\sf c}(z)$, $\mathscr C_{2}^{\sf c}(z_1,z_2)=\mathscr C_{2}(z_1,z_2)-\mathscr C_{1}(z_1)\mathscr C_1(z_2)$,
\begin{align}
\nonumber
\mathscr C_{3}^{\sf c}(z_1,z_2,z_3)={}&\mathscr C_{3}(z_1,z_2,z_3)-\mathscr C_{2}(z_1,z_2)\mathscr C_1(z_3)-\mathscr C_{2}(z_2,z_3)\mathscr C_1(z_1)
\\
&-\mathscr C_{2}(z_1,z_3)\mathscr C_1(z_2)+2\,\mathscr C_{1}(z_1)\mathscr C_1(z_2)\mathscr C_1(z_3).
\end{align}
We now express the connected cumulant functions in terms of the monic orthogonal polynomials $P_\ell(z)=z^\ell+\dots$ uniquely defined by
\be \label{monicOPintro}
\int_I P_\ell(x)P_m(x)\e^{V(x)}\d x=h_\ell\delta_{\ell,m},
\ee
and of the $2\times 2$ matrix
\be \label{Ymatrix}
Y_N(z):=
\renewcommand{\arraystretch}{1.25}\left(\begin{array}{cc}
P_N(z) & \frac 1{2\pi\i}\int_IP_N(x)\e^{V(x)}\frac{\d x}{x-z}
\\
-\frac{2\pi\i}{h_{N-1}}P_{N-1}(z) & -\frac{1}{h_{N-1}}\int_IP_{N-1}(x)\e^{V(x)}\frac{\d x}{x-z}
\end{array}\right),
\ee
which is the well-known solution to the \emph{Riemann--Hilbert problem} of orthogonal polynomials \cite{FIK1992}; it is an analytic function of $z\in\C\setminus I$.

\begin{shaded}
\begin{theorem}\label{thmcumulants}
Let \be
R(z):=Y_N(z)\begin{pmatrix}1 & 0\\ 0 &0 \end{pmatrix}Y_N^{-1}(z)\,,
\ee
with $Y_N(z)$ as in \eqref{Ymatrix}. Then the  connected  cumulant functions \eqref{disctoconn} are given by
\begin{align}
\mathscr C_1^{\sf c}(z)&=\left(Y^{-1}_N(z)Y_N'(z)\right)_{1,1},
\label{OnePointThm}
\\
\mathscr C_2^{\sf c}(z_1,z_2)&=\frac{\tr\left(R(z_1)R(z_2)\right)-1}{(z_1-z_2)^2},		\label{TwoPointThm}
\\
\mathscr C_\ell^{\sf c}(z_1,\dots,z_\ell)&=-\sum_{(i_1,\dots,i_\ell)\in\cyc( (\ell) )}\frac{\tr\left(R(z_{i_1})\dots R(z_{i_\ell})\right)}{(z_{i_1}-z_{i_2})\cdots(z_{i_\ell}-z_{i_1})},\quad\ell\geq 3, 		\label{MultiPointThm}
\end{align}
where prime in the second  formula denotes derivative with respect to $z$ and  $\cyc((\ell))$ in the last formula is the set of $\ell$-cycles in the symmetric group $\S_\ell$.
\end{theorem}
\end{shaded}

The proof is given in Section \ref{sec3}. 
Formul\ae\ of this sort for correlators of hermitian matrix models have been recently discussed in the literature, see e.g.~\cite{EKR2015,DY2017}.
They are directly related to the theory of tau functions (formal \cite{DY2017} and isomonodromic \cite{BEH2006,GGR2020}) and to topological recursion theory \cite{CE2006,EO2007,BE2009,BBE2015}.
Incidentally, similar formul\ae\ also appear for matrix models with external source \cite{K1992,BC2017,BDY2016,BDY2018,BR2019}.
In Section~\ref{sec3} we provide a direct derivation based on the Riemann--Hilbert characterization of the matrix~$Y_N(z)$.

We can apply these formul\ae\ to the Jacobi measure $\d m_N^{\sf J}$, see \eqref{JUEmeasure}, for which the support is $I=[0,1]$.
Therefore we can expand the cumulants near the points $z=0$ or $z=\infty$; the expansion at $z=1$ could be considered but we omit it as it is recovered from the one at $z=0$ by exchanging $\a,\b$, see \eqref{JUEmeasure}.
Using the definition in \eqref{cumulantfunctions} and \eqref{disctoconn}, we obtain the generating functions for the JUE connected correlators \eqref{correlatorsintro}, namely
\be
\mathscr C_1(z)\overset{z\to \infty}{\sim} \mathscr F_{1,\infty}^{\sf c}(z)-\frac{N}{z},\quad
\mathscr C_1(z)\overset{z\to 0}{\sim}\mathscr F_{1,0}^{\sf c}(z),\quad 
\mathscr C_\ell^{\sf c}(z_1,\dots,z_\ell)\overset{z\to p}{\sim}\mathscr F_{\ell,p}^{\sf c}(z_1,\dots,z_\ell)\quad (p=0,\infty),
\ee
where 
\begin{align}
\nonumber
\mathscr F_{\ell,\infty}^{\sf c}(z_1,\dots z_\ell)&:=\sum_{k_1,\dots,k_\ell\geq 1}\frac{\left\langle\prod_{j=1}^\ell\tr X^{k_j}\right\rangle^{\sf c}}{z_1^{k_1+1}\cdots z_\ell^{k_\ell+1}},
\\
\label{cumulantformal}
\mathscr F_{\ell,0}^{\sf c}(z_1,\dots,z_\ell)&:=(-1)^\ell\sum_{k_1,\dots,k_\ell\geq 1}\left\langle\prod_{j=1}^\ell\tr X^{-k_j}\right\rangle^{\sf c} z_1^{k_1-1}\cdots z_\ell^{k_\ell-1}.
\end{align}
On the other hand, performing the same expansion on the right hand side of the expressions for the cumulants in Theorem~\ref{thmcumulants}, we have an explicit tool  to compute the correlators.

For the specific case of Jacobi polynomials, we  prove in Section \ref{sec4} (Proposition \ref{Rexpansions}) that at  $z=\infty$  the matrix  $R(z)$ has the Taylor expansion  (valid for $|z|>1$)  of the form   $R(z) =T^{-1}R^{[\infty]}(z)T$
and the Poincar\'e asymptotic expansion $R(z) \sim T^{-1} R^{[0]}(z)T $ at $z=0$  valid  in the sector $0<\arg z<2\pi$. Here  $T$ is the constant matrix 
\be
\label{Tmatrix}
T=
\begin{pmatrix}
1 & 0 \\
 0 & \frac{h_{N-1}^{\sf J}}{2 \pi i} \frac{1}{ (\a+\b+2N)(\a+\b+2N-1)}
\end{pmatrix},
\ee
with $h_\ell^{\sf J}$ given in  \eqref{NormingC}, and the series $R^{[\infty]}(z),R^{[0]}(z)$ are
\begin{align}
R^{[\infty]}(z)&=\begin{pmatrix} 1 & 0 \\ 0 & 0 \end{pmatrix}+\sum_{\ell\geq 0}\frac 1{z^{\ell+1}}
\frac{1}{\a+\b+2N}\begin{pmatrix}
\ell A_{\ell}(N) & N (\a+N) (\b+N) (\a+\b+N)  B_\ell(N+1)
\\
-  B_\ell(N) & -\ell A_{\ell}(N)
\end{pmatrix},\nonumber
\\
\label{Rinftyintro}
\\
R^{[0]}(z)&=\begin{pmatrix} 1 & 0 \\ 0 & 0 \end{pmatrix}+\sum_{\ell\geq 0}
\frac{z^\ell}{\a+\b+2N}\begin{pmatrix}
(\ell+1) \wt A_{\ell}(N) & -N (\a+N) (\b+N) (\a+\b+N) \wt B_\ell(N+1)
\\
\wt B_\ell(N) & -(\ell+1) \wt A_{\ell}(N)
\end{pmatrix},
\label{Rzerointro}
\end{align}
where
\begin{align}
\nonumber
A_0(N)&=
 \frac{N(\b+N)}{\a+\b+2N}, & 
\\	 
\nonumber
 A_\ell(N)&=\frac{N (\a+N) (\b+N) (\a+\b+N) (\a+2)_{\ell-1}}{(\a+\b+2 N-1)_{\ell+2}} \pFq{4}{3}{1-\ell,\ell+2,1-\b-N,1-N}{2,\a+2,2-\a-\b-2 N}{1}, & \ell\geq 1,
\\
\label{AsBs}
B_\ell(N)&= \frac{(\a+1)_\ell}{(\a+\b+2 N-1)_{\ell+1}} \pFq{4}{3}{-\ell,\ell+1,1-\b-N,1-N}{1,\a+1,2-\a-\b-2 N}{1},&\ell\geq 0,
\end{align}
and
\be
\wt A_\ell(N)=\frac{(\a+\b+2N-\ell)_{2\ell+1}}{(\a-\ell)_{2\ell+1}}A_\ell(N),
\quad
\wt B_\ell(N)=\frac{(\a+\b+2N-1-\ell)_{2\ell+1}}{(\a-\ell)_{2\ell+1}} B_\ell(N),\quad\ell\geq 0.
\label{AsBstilde}
\ee
Here ${}_4F_3$ is the generalized hypergeometric function, and we use the \emph{rising factorial}
\be
\label{risingfactorial}
(s)_k=s(s+1)\cdots(s+k-1).
\ee
For example, the first few terms read
\begin{align}
\nonumber
A_1(N)={}& \frac{N (\a+N) (\b+N) (\alpha +\beta+N)}{(\alpha +\beta+2N -1) (\alpha +\beta+2N) (\alpha +\beta+2N +1)}, \\
\nonumber
B_0(N)={}& \frac{1}{\alpha +\beta +2 N-1}, \\
B_1(N)={}& \frac{ (\a -1) (\alpha +\beta )+2N (\alpha +\beta -1)+2N^2}{(\alpha +\beta+2N -2) (\alpha +\beta+2N -1) (\alpha +\beta+2 N )}.
\end{align}
Since the constant conjugation by $T$ in \eqref{Tmatrix} of the matrix $R(z)$ does not affect the formul\ae\ of Theorem \ref{thmcumulants} (see also Section \ref{sec4}) we obtain the following corollary, which provides explicit formul\ae\ for the generating functions  of the correlators.

\begin{shaded}
\begin{corollary}
\label{corollaryWilson}
Let $R^{[\infty]}(z)$ and $R^{[0]}(z)$ be the explicit series given in \eqref{Rinftyintro} and \eqref{Rzerointro}. The one-point generating function \eqref{cumulantformal} of the JUE are
\begin{align}
\mathscr F_{1,\infty}(z)&=\frac{\a+\b+2N}{z(1-z)}\int_\infty^z\left(1-R_{1,1}^{[\infty]}(w)\right)\d w-\frac{N(\a+N)}{z(1-z)(\a+\b+2N)},
\\
\mathscr F_{1,0}(z)&=\frac{\a+\b+2N}{z(1-z)}\int_0^z\left(1-R_{1,1}^{[0]}(w)\right)\d w-\frac{N}{1-z},
\intertext{where $R_{1,1}^{[\infty]},R_{1,1}^{[0]}$ denote the $(1,1)$-entry of $R^{[\infty]},R^{[0]}$ respectively. The multi-point generating functions \eqref{cumulantformal} are}
\mathscr F_{2,p}^{\sf c}(z_1,z_2)&=\frac{\tr\left(R^{[p]}(z_1)R^{[p]}(z_2)\right)-1}{(z_1-z_2)^2},
\\
\mathscr F_{\ell,p}^{\sf c}(z_1,\dots,z_\ell)&=-\sum_{(i_1,\dots,i_\ell)\in\cyc((\ell))}\frac{\tr\left(R^{[p]}(z_{i_1})\dots R^{[p]}(z_{i_\ell})\right)}{(z_{i_1}-z_{i_2})\cdots(z_{i_\ell}-z_{i_1})},\quad\ell\geq 3,\qquad p=0,\infty.			
\end{align}
\end{corollary}
\end{shaded}

The proof is in Section \ref{sec42}, and is obtained from the formul\ae\ of Theorem \ref{thmcumulants} by expansion at $z_i\to\infty,0$.
In this corollary, the formul\ae\ on the right hand side are interpreted as power series expansions at $z=0$ or $z=\infty$. 
To this end, we remark that for $\ell\geq 2$ these series are well defined, as it follows from the fact that the corresponding analytic functions are holomorphic in $(\C\setminus I)^\ell$ and in particular regular along the diagonals $z_a=z_b$ for $a\not=b$.

The coefficients of $R^{[0]}(z)$ and $R^{[\infty]}(z)$ are rational functions of $N,\a,\b$ and we conclude that JUE correlators extend to rational functions of $N,\a,\b$.

Examining more closely the formula for $\mathscr F_{1,\infty}$ we see that
\be
(1-z)\mathscr F_{1,\infty}(z)=-\frac Nz+\sum_{k\geq 0}\frac{\left\langle\tr X^k\right\rangle-\left\langle\tr X^{k+1}\right\rangle}{z^{k+1}},
\ee
which by the explicit expansion $R^{[\infty]}(z)$ in \eqref{Rinftyintro} implies
\be
\label{onepointdiff+}
\left\langle\tr X^k\right\rangle-\left\langle\tr X^{k+1}\right\rangle=A_k(N),
\ee
where $A_k(N)$ is defined in \eqref{AsBs}.
Reasoning in the same way for $\mathscr F_{1,0}(z)$ we obtain
\be
\label{onepointdiff-}
\left\langle\tr X^{-k-1}\right\rangle-\left\langle\tr X^{-k}\right\rangle=\wt A_k (N) = \frac{(\a+\b+2N-k)_{2k+1}}{(\a-k)_{2k+1}} \left( \left\langle\tr X^k\right\rangle-\left\langle\tr X^{k+1}\right\rangle \right),
\ee
where $\wt A_k(N)$ is defined in \eqref{AsBstilde}.
Equations \eqref{onepointdiff+} and \eqref{onepointdiff-} agree with the results of \cite{CMOS2019}.

\begin{remark}
The coefficients $A_\ell(N),B_\ell(N)$ can be expressed in terms of Wilson polynomials \cite{Wilson,KS1994}, which are defined  by
\be
\frac{W_n(k^2;a,b,c,d)}{(a+b)_n (a+c)_n (a+d)_n}= \pFq{4}{3}{-n,n+a+b+c+d-1,a+ \i k, a- \i k}{a+b,a+c,a+d}{1}, \label{HyperW}
\ee
for more details see Proposition \ref{Wilson}.
Thus the formul\ae\ of Corollary \ref{cormain} extend the connection between JUE moments $\left\langle\tr X^k\right\rangle$ and Wilson polynomials described in \cite{CMOS2019} to the JUE multi-point correlators $\left\langle\tr X^{k_1}\cdots\tr X^{k_\ell}\right\rangle$.
\end{remark}

\begin{remark}[JUE mixed correlators]\label{examplemixedcorrelators}
We could consider more general generating functions as follows; take $q,r,s\geq 0$ with $q+r+s>0$ and expand the cumulant function
\be \label{mixedCorrelators}
\mathscr C_{q+r+s}(z_1,\dots,z_q,w_1,\dots,w_r,y_1,\dots,y_{s})
\ee
for the Jacobi measure as $z_i\to\infty,w_i\to 0,y_i\to 1$, to obtain the generating function
\begin{align}
\nonumber
\sum_{\begin{smallmatrix}k_1,\dots,k_q\geq 1\\i_1,\dots,i_r\geq 1\\j_1,\dots,j_s\geq 1\end{smallmatrix}}\int_{\H_N(0,1)}\tr X^{k_1}\cdots\tr X^{k_q}\tr X^{-i_1}\cdots\tr X^{-i_r}\tr(\1-X)^{j_1}\cdots\tr(\1-X)^{j_s}\d m_N^{\sf J}(X)
\\
\qquad\qquad\qquad\qquad\qquad\times\frac{w_1^{i_1-1}\cdots w_r^{i_r-1}(y_1-1)^{j_1-1}\cdots(y_s-1)^{j_s-1}}{z_1^{k_1+1}\cdots z_q^{k_q+1}}.
\end{align}
It is then clear that we can compute the coefficients of such series in terms of the matrix series $R^{[0]},R^{[\infty]}$, and thus of Wilson polynomials, by the formul\ae\ of Theorem \ref{thmcumulants}; note that the expansion of $R(z)$ at $z=1$ is obtained from $R^{[0]}$ by exchanging $\a$ with $\b$ and $z$ with $1-z$.
\end{remark}

\begin{example}\label{examplecomputation}
From the formul\ae\ of Corollary \ref{corollaryWilson} we can compute
\be
\left \langle (\tr X)^3 \right\rangle^{\sf c} = \frac{2N(\a+\b)(\b-\a)(\a+N)(\b+N)(\a+\b+N)}{(\a+\b+2N-2)(\a+\b+2N-1)(\a+\b+2N)^3(\a+\b+2N+1)(\a+\b+2N+2)}.
\ee
With the substitution $\a=(c_\a-1)N$ and $\b=(c_\b-1)N$ we have the large $N$ expansion
\begin{align}
\nonumber
\left \langle \left(\tr X\right)^3  \right\rangle^{\sf c} ={}&  \frac{1}{N}\left[ c_\a \left(\frac{2}{(c_\a+c_\b)^3}-\frac{6}{(c_\a+c_\b)^4}+\frac{4}{(c_\a+c_\b)^5}\right) \right.
\\ 
\nonumber
&\ \
+c_\a^2 \left(-\frac{6}{(c_\a+c_\b)^4}+\frac{18}{(c_\a+c_\b)^5}-\frac{12}{(c_\a+c_\b)^6}\right)
\\
&\ \
\left. + c_\a^3 \left(\frac{4}{(c_\a+c_\b)^5}-\frac{12}{(c_\a+c_\b)^6}+\frac{8}{(c_\a+c_\b)^7}\right) \right] +\mathcal{O}\left(\frac{1}{N^3}\right).
\end{align}
Matching the coefficients as in Theorem \ref{thmhurwitz} we get the values for $h_{g=0}^{\sf c}(\ll=(1,1,1),\mu,\nu)$ (the connected Hurwitz numbers defined in Remark \ref{ConnectedH}) reported in the following table;
\be
\begin{array}[t]{c|c|c|c}
& \nu=(3) & \nu=(2,1) & \nu=(1,1,1)\\
\hline \mu=(3) & 2 & 6 & 4\\
\hline \mu=(2,1) & 6 & 18 & 12 \\
\hline \mu=(1,1,1) & 4 & 12 & 8\\
\end{array}
\ee
For example, the numbers in the first row ($\mu=(3)$) can be read from the following factorizations in $\mathfrak S_3$. To list them let us first note that we have $\cyc(\ll)=\{\Id\}$ and $\cyc(\mu)=\{(123),(132)\}$; therefore for $\nu=(3)$ we have 2 factorizations ($r=\mbox{number of transpositions}=0$)
\be
(123) (132) = \Id, \qquad (132) (123) = \Id,
\ee
for $\nu=(2,1)$ ($\cyc(\nu)=\{(12),(23),(13)\}$) we have 6 factorizations ($r=1$)
\begin{align}
&(123)(12)(13)=\Id,	& 	&(123)(13)(23)=\Id,	&	&(123)(23)(12)=\Id,\\
&(132)(13)(12)=\Id,	& 	&(132)(12)(23)=\Id,	&	&(132)(23)(13)=\Id,
\end{align}
and for $\nu=(1,1,1)$ we have the 4 factorizations ($r=2$, here the monotone condition plays a role)
\begin{align}
&(123) \Id (12)(13) = \Id,	&	&(123) \Id (13)(23) = \Id,	\\
&(132) \Id (12)(23) = \Id,	&	&(123) \Id (23)(13) = \Id.
\end{align}
Similarly we can compute from Corollary \ref{corollaryWilson}
\begin{align}
\nonumber
\left\langle\left(\tr X^{-1}\right)^3\right\rangle^{\sf c}&=
\frac{2 N (\a+N) (\a+2 N) (\b+N) (\a+\b+N) (\a+2\b+2N)}{(\a-2) (\a-1) \a^3 (\a+1) (\a+2)}
\\
\nonumber
&=\frac 1N\left[
\left(\frac{2}{(c_\a-1)^3}+\frac{6}{(c_\a-1)^4}+\frac{4}{(c_\a-1)^5}\right) (c_\a+c_\b-1)\right.
\\
\nonumber
&\qquad-\left(\frac{6}{(c_\a-1)^4}+\frac{18}{(c_\a-1)^5}+\frac{12}{(c_\a-1)^6}\right) (c_\a+c_\b-1)^2
\\
&\qquad
\left.+\left(\frac{4}{(c_\a-1)^5}+\frac{12}{(c_\a-1)^6}+\frac{8}{(c_\a-1)^7}\right) (c_\a+c_\b-1)^3
\right]+\O\left(\frac 1{N^3}\right)
\end{align}
and from Theorem \ref{thmhurwitz} we recognize the connected Hurwitz numbers tabulated above.
\end{example}

\begin{remark}[Laguerre limit] \label{LaguerreLimit}
There is a scaling limit of the JUE correlators to the LUE correlators; if $k_1,\cdots,k_\ell$ are arbitrary integers we have
\be\label{LUElimit}
\lim_{\b\to+\infty}\b^{k_1+\cdots+k_\ell}\left\langle\prod_{j=1}^\ell\tr X^{k_j}\right\rangle=\frac{
\int_{\H_N(0,+\infty)}\left(\prod_{j=1}^\ell\tr X^{k_j}\right){\det}^\a(X)\exp(-\tr X)\d X}{\int_{\H_N(0,+\infty)}{\det}^\a(X)\exp(-\tr X)\d X}.
\ee
Therefore the results of the present work about the JUE directly imply analogous results for the LUE; these results are already known from \cite{CDO2018,GGR2020}. 
See also Remark \ref{remLUE2}.
\end{remark}

\section{JUE Correlators and Hurwitz Numbers}\label{sec2}

In this section we prove Theorem \ref{thmhurwitz}. For the proof we will consider the so-called \emph{multiparametric weighted Hurwitz numbers}; this far-reaching generalization of classical Hurwitz numbers was introduced and related to tau functions of integrable systems in several works by Harnad, Orlov  \cite{HO2015}, and Guay-Paquet \cite{GPH2015}, after the impetus of the seminal work of Okounkov \cite{O2000}.

\subsection{Multiparametric weighted Hurwitz numbers}\label{secHurwitz}

Let $\C[\S_n]$ be the group algebra of the symmetric group $\S_n$; namely, $\C[\S_n]$ consists of formal linear combinations with complex coefficients of permutations of $\{1,\dots,n\}$. We shall need two important type of elements of $\C[\S_n]$, which we now introduce. 
For any $\ll\vdash n$ denote
\be
\label{cycleelements}
\CC_\ll:=\sum_{\pi\in\cyc(\ll)}\pi,
\ee
where we recall that $\cyc(\ll)\subset\S_n$ is the conjugacy class of permutations of cycle-type $\ll$. It is well known \cite{Serre} that the set of $\CC_\ll$ for $\ll\vdash n$ form a linear basis of the center $Z(\C[\S_n])$ of the group algebra.

The second class of elements consists of the \emph{Young--Jucys--Murphy} (YJM) elements \cite{J1974,M1981} $\J_a$, for $a=1,\dots,n$, defined as
\be
\J_1=0,\quad \J_a=(1,a)+(2,a)+\dots+(a-1,a),\ 2\leq a\leq n,
\ee
denoting $(a,b)$ (with $a<b$) the transposition of $\{1,\dots,n\}$ switching $a,b$ and fixing everything else.

Although singularly the YJM elements are not central, they commute amongst themselves, and symmetric polynomials of $n$ variables evaluated at $\J_1,\dots,\J_n$ generate $Z(\C[\S_n])$. Indeed the following relation \cite{J1974} takes place in $Z(\C[\S_n])[\epsilon]$;
\be
\label{JMvsC}
\prod_{a=1}^n(1+\epsilon\J_a)=\sum_{\ll\vdash n}\epsilon^{n-\ell(\ll)}\CC_\ll.
\ee

With these preliminaries we are ready to introduce the class of multiparametric Hurwitz numbers \cite{HO2015,GPH2015,BHR2020} which we need. Fix the real parameters $\gamma_1,\dots,\gamma_L$ and $\delta_1,\dots,\delta_M$ ($L,M\geq 0$) and collect them into the rational function
\be
\label{eq:G}
G(z):=\frac{\prod_{i=1}^L(1+\gamma_iz)}{\prod_{j=1}^M(1-\delta_jz)}.
\ee
Then, the \emph{(rationally weighted) multiparametric (single) Hurwitz numbers} $H_G^{d}(\ll)$, associated  to the function $G$  in \eqref{eq:G} and labeled by the integer $d\geq 1$ and by the partition $\ll\vdash n$, are defined by
\be
\label{eq:defhurwitz}
H_G^{d}(\ll):=\frac 1{z_\ll}[\epsilon^d \CC_\ll]\prod_{a=1}^nG\left(\epsilon\J_a\right),
\ee
where this notation $[\epsilon^d \CC_\ll]$  denotes the coefficient in front of $\epsilon^d\CC_\ll$ in the expansion of $\prod_{a=1}^nG\left(\epsilon\J_a\right)\in Z(\C[\S_n])[[\epsilon]]$ in the basis $\{\CC_\ll\}$; to compute the expression $G\left(\epsilon\J_a\right)\in Z(\C[\S_n])[[\epsilon]]$, the denominators in \eqref{eq:G} are to be understood as $(1-\delta_jz)^{-1}=\sum_{k\geq 0}\delta_j^kz^k$.

\subsection{Generating functions of multiparametric Hurwitz numbers in the Schur basis}

The following result (see \cite{HO2015}) expresses the generating functions of multiparametric weighted Hurwitz numbers in the \emph{Schur basis}. 
In this context, the latter is regarded as the basis $\{s_\ll(\t)\}$ ($\ll$ running in the set of all partitions) of the space of weighted homogeneous polynomials in $\t=(t_1,t_2,\dots)$, with $\deg t_k=k$, whose elements are
\be
\label{schur}
s_\ll(\t)=\det\left[h_{\ll_i-i+j}(\t)\right]_{i,j=1}^{\ell(\ll)},
\ee
where the complete homogeneous symmetric polynomials $h_k(\t)$ are defined by the generating series\footnote{For convenience we adopt a normalization which differs from the one common in the literature by a transformation $t_k\mapsto t_k/k$.}
\be
\sum_{k\geq 0}w^kh_k(\t)=\exp\left(\sum_{k\geq 1}\frac{t_k}kw^k\right).
\ee

In the following we shall denote $\mathcal P$ the set of all partitions.

\begin{proposition}[\cite{HO2015}]
\label{prop}
The generating function
\be
\label{tauGmonomial}
\tau_G(\epsilon;\t)=\sum_{d\geq 1}\epsilon^d\sum_{\ll\in\P}H_G^d(\ll)\prod_{i=1}^{\ell(\ll)}t_{\ll_i}
\ee
of multiparametric weighted Hurwitz numbers \eqref{eq:defhurwitz} associated to the rational function \eqref{eq:G} is equivalently expressed as
\be
\label{tauGschur}
\tau_G(\epsilon;\t)=\sum_{\ll\in\P}\frac{\dim\ll}{|\ll|!}r^{(G,\epsilon)}_\ll s_\ll(\t),
\ee
where $s_\ll(\t)$ are the Schur polynomials \eqref{schur} and the coefficients are given explicitly by
\be
\label{rlambdaG}
r^{(G,\epsilon)}_\ll=\prod_{(i,j)\in\ll}G(\epsilon(j-i)),
\ee
$\dim\ll$ being the dimension of the irreducible representation of $\S_{|\ll|}$ associated with $\ll$.
\end{proposition}

Before the proof we give a couple of remarks.
\begin{enumerate}
\item In \eqref{rlambdaG} and below we use the notation $(i,j)\in\ll$ where the partition $\ll$ is identified with its \emph{diagram}, i.e.~the set of $(i,j)\in\mathbb Z^2$ satisfying $1\leq i\leq\ell(\ll)$, $1\leq j\leq\ll_i$.
For example, the diagram of the partition $\ll=(4,2,2,1)\vdash 9$ is depicted below;
\be
\label{diagram}
\begin{array}{c|cccc}
 & j=1 & j=2 & j=3 & j=4
\\
\hline
i=1 & \bullet&\bullet&\bullet&\bullet
\\
i=2 & \bullet&\bullet&&
\\
i=3 & \bullet&\bullet&&
\\
i=4 & \bullet&&&
\end{array}
\ee
\item There exist several equivalent formul\ae\  for $\dim\ll$, including the well-known hook-length formula; for later convenience we recall the expression
\be
\label{dimll}
\frac{\dim\ll}{|\ll|!}=
\frac{\prod_{1 \leq i < j \leq N}(\lambda_i-\lambda_j+j-i)}
{\prod_{k=1}^N(\lambda_k-k+N)!},
\ee
valid for all $N\geq\ell(\ll)$, setting $\ll_i=0$ for all $\ell(\ll)<i\leq N$.
\end{enumerate}

\begin{proof}[Proof of Proposition~\ref{prop}.]
We need a few preliminaries. First we recall that $Z\left(\C[\S_n]\right)$ is a semi-simple commutative algebra; a basis of idempotents is given by (see e.g. \cite{Serre}) 
\be
\label{idempotents}
\mathcal E_\ll=\frac{\dim\ll}{|\ll|!}\sum_{\mu\vdash n}\chi_\ll^\mu\CC_\mu,
\ee
where $\chi_\ll^\mu$ are the characters of the symmetric group and $\CC_\mu$ are given in \eqref{cycleelements}.
Namely
\be
\label{projection}
\mathcal E_\ll\mathcal E_{\ll'}=\begin{cases}
\mathcal E_\ll&\ll=\ll'\\0&\ll\not=\ll'.
\end{cases}
\ee
For any symmetric polynomial $p(y_1,\dots,y_n)$ in $n$ variables, $p(\J_1,\dots,\J_n)$ belongs to $Z\left(\C[\S_n]\right)$, as we have already mentioned; central elements are diagonal on the basis of idempotents and it is proven in \cite{J1974} that
\be
\label{eigenvalues}
p(\J_1,\dots,\J_n)\mathcal E_\ll=p\left(\{j-i\}_{(i,j)\in\ll}\right)\mathcal E_\ll,
\ee
where in the right hand side we denote $p\left(\{j-i\}_{(i,j)\in\ll}\right)$ the evaluation of the symmetric polynomial $p$ at the $n$ values of $j-i$ for $(i,j)\in\mathbb Z^2$ in the diagram of $\ll\vdash n$; in the example $\ll=(4,2,2,1)\vdash 9$ above, see \eqref{diagram}, this denotes the evaluation $p(0,1,2,3,-1,0,-2,-1,-3)$.

We are ready for the proof proper. First note that by \eqref{eigenvalues} and \eqref{rlambdaG} we have
\be
\left[\prod_{a=1}^nG(\epsilon\J_a)\right]\mathcal E_\ll=r_\ll^{(\epsilon,G)}\mathcal E_\ll,
\ee
which implies, using \eqref{projection}, that
\be
\prod_{a=1}^nG(\epsilon\J_a)=\sum_{\ll\vdash n}r_\ll^{(\epsilon,G)}\mathcal E_\ll.
\ee
By the definition of $H_G^d(\mu)$  in \eqref{eq:defhurwitz} we can rewrite the last identity as
\be
\sum_{\mu\vdash n}\sum_{d\geq 1}\epsilon^d z_\mu H_G^d(\mu)\CC_\mu=\sum_{\ll\vdash n}r_\ll^{(\epsilon,G)}\mathcal E_\ll=\sum_{\ll,\mu\vdash n}\frac{\dim\ll}{|\ll|!}r_\ll^{(\epsilon,G)}\chi_\ll^\mu\CC_\mu.
\ee
Since $\CC_\mu$ form a basis of $Z\left(\C[\S_n]\right)$ we get that for any partition $\mu$
\be
\sum_{d\geq 1}\epsilon^d H_G^d(\mu)=\sum_{\ll\vdash |\mu|}\frac{\dim\ll}{|\ll|!}r_\ll^{(\epsilon,G)}\frac{\chi_\ll^\mu}{z_\mu}.
\ee
Multiplying this identity by $\prod_{i=1}^{\ell(\mu)}t_{\mu_i}$ and summing over all partitions $\mu$, on the left we obtain \eqref{tauGmonomial} and on the right, thanks to the well-known identity \cite{Macdonald}
\be
s_\ll(\t)=\sum_{\mu\vdash |\ll|}\frac{\chi^\mu_\ll}{z_\mu}\prod_{i=1}^{\ell(\mu)}t_{\mu_i},
\ee
we obtain \eqref{tauGschur}. The proof is complete.
\end{proof}

\begin{remark}
This result is used by the authors of \cite{HO2015} to prove that the generating function $\tau_G(\epsilon;\t)$ is a one-parameter family in $\epsilon$ of \emph{Kadomtsev--Petviashvili tau functions} in the times $\t$; a tau function such that  the coefficients of the Schur expansion have the form \eqref{rlambdaG} is termed \emph{hypergeometric} tau function.
It is also worth remarking that the theorem stated here is a reduction of a more general result, proved in \cite{GPH2015}, dealing with generating functions of \emph{double} (weighted) Hurwitz numbers.
In this general setting, the corresponding integrable hierarchy is the \emph{2D Toda hierarchy}.
\end{remark}

\subsection{JUE partition functions}\label{par22}

Let us introduce the formal generating functions
\be
\label{Zpm}
Z_N^\pm(\u):=\int_{\H_N(0,1)}\exp\left(\sum_{k\geq 1}\frac{u_k}k\tr X^{\pm k}\right)\d m^{\sf J}_N(X)=\sum_{\ll\in\P}\frac{\left\langle\prod_{j=1}^\ell\tr X^{\pm \ll_j}\right\rangle}{z_\ll}\prod_{i=1}^{\ell(\ll)}u_{\ll_i},
\ee
of JUE correlators; the sum in the right hand side is a formal power series in $\bf u$ running over all partitions $\ll$, with the combinatorial factor $z_\ll$ defined in \eqref{zll}\footnote{A formal power series in infinitely many variables $u_1,u_2,\cdots$ can be rigorously treated by introducing the grading $\deg u_k=k$ and working in the completion of the algebra of polynomials in $u_1,u_2,\dots$, filtered by degree.}.
We call $Z^+_N(\u)$ (resp. $Z^-_N(\u)$) the \emph{positive} (resp.~\emph{negative}) JUE partition function.
Although it will not be needed in the following, we mention that these partition functions are Toda tau functions in the times $u_1,u_2,\dots$ \cite{AvM1995,AvM2001,Deift}.

Our goal in this paragraph is to show that the JUE partition functions can be expressed in the form \eqref{tauGschur} for appropriate choice of $G$ (see Corollary \ref{cormain}).

The first step is to expand the JUE partition functions in the Schur basis; this is achieved by the following well-known general lemma, whose proof we report for the reader's convenience.
The idea of expanding a hermitian matrix model partition function over the Schur basis has been recently used in the computation of correlators in \cite{JKM2020}.

We first introduce the following notations
\be
\Delta(\underline x)=\prod_{1\leq i< j\leq N}(x_j-x_i)=\det\left(x^{N-j}_i\right)_{i,j=1}^N
\ee
for the Vandermonde determinant and
\be
\label{characters}
\chi_\ll(\underline x):=\frac{\det\left[x^{N-i+\ll_i}_j\right]_{i,j=1}^N}{\Delta(\underline x)}
\ee
for the characters of ${\rm GL}_n$; again, we set $\ll_i=0$ for all $\ell(\ll)<i\leq N$.
\begin{lemma}
\label{lemmamatrixmodelschur}
For any potential $V(x)$ ($x\in I$) we have
\be
\frac{\int_{\H_N(I)}\exp\tr\left(V(X)+\sum_{k\geq 1}\frac{u_k}kX^{\pm k}\right)\d X}{\int_{\H_N(I)}\exp\tr\left(V(X)\right)\d X}
=\sum_{\ll\in\P:\,\ell(\ll)\leq N}c_{\ll,N}^{\pm}s_\ll({\bf u}),
\ee
where the Schur polynomials are defined in \eqref{schur} and the coefficients are
\be
\label{coefficients1}
c_{\ll,N}^\pm=\frac{\int_{I^N} \chi_\lambda(\underline x^{\pm 1}) \Delta^2(\underline x) \prod_{a=1}^N\exp \left[V(x_a)\right]\d^N\underline x}{\int_{I^N}\Delta^2(\underline x)\prod_{a=1}^N\exp \left[V(x_a)\right]\d^N \underline x}.
\ee
Here $\underline x=(x_1,\dots,x_N)$ and $\underline x^{-1}=(x_1^{-1},\dots,x_N^{-1})$.
\end{lemma}

\begin{proof}
We have
\be \label{IntReduction}
\frac{\int_{\H_N(I)}\exp\tr\left(V(X)+\sum_{k\geq 1}\frac{u_k}kX^{\pm k}\right)\d X}{\int_{H_N(I)}\exp\tr\left(V(X)\right)\d X}
=\frac{\int_{I^N}\Delta^2(\underline x)\prod_{a=1}^N\exp\left[V(x_a)+\sum_{k\geq 1}\frac{u_k}kx_a^{\pm k}\right]\d^N \underline x}{\int_{I^N}\Delta^2(\underline x)\prod_{a=1}^N\exp \left[V(x_a)\right]\d^N \underline x},
\ee
where we use the standard decomposition $\d X=\Delta^2(\underline x)\d^N\underline x\d U$ of the Lebesgue measure into eigenvalues $\underline x=(x_1,\dots,x_N)$ and eigenvectors $U\in\U_N$ of the hermitian matrix $X=UXU^\dagger$, with $\d U$ a Haar measure on $\U_N$ (whose normalization is irrelevant as it cancels in \eqref{IntReduction} between numerator and denominator). 
The proof follows by an application of the identity
\be
\exp\left[\sum_{k\geq 1}\frac{u_k}k(x_1^{\pm 1}+\dots+x_N^{\pm 1})^k\right]=\sum_{\ll\in\P:\,\ell(\ll)\leq N}\chi_\ll(\underline x^{\pm 1})s_\ll(\u),
\ee
which is nothing but a form of Cauchy identity, see e.g. \cite{Stanley}.
\end{proof}

\begin{remark}
By applying Andr\'eief identity
\be
\int_{I^N}\det\left[f_i(x_j)\right]_{i,j=1}^N\det\left[g_i(x_j)\right]_{i,j=1}^N\d^N\underline x=N!\det\left[\int_I f_i(x)g_j(x)\d x\right]_{i,j=1}^N
\ee
it is straightforward to show that the coefficients $c_{\ll,N}$ in \eqref{coefficients1} can also be expressed as
\be
c_{\ll,N}^\pm=\frac{\det\left[\mathcal M^\pm_{\ll_i+N-i,N-j}\right]_{i,j=1}^N}{\det\left[\mathcal M^\pm_{N-i,N-j}\right]_{i,j=1}^N},
\qquad
\mathcal M_{i,j}^\pm=\int_I x^{\pm(i+j)}\e^{V(x)}\d x,
\ee
see also~\cite{JKM2020}.
However, for our purposes it is more convenient to work with the representation \eqref{coefficients1}.
\end{remark}

Applying this general lemma to $I=[0,1]$ and $V(x)=\a\log x+\b\log(1-x)$ we can expand the positive and negative JUE partition functions in the Schur basis as
\be
\label{schurexpansion}
Z_N^{\pm}(\u)=\sum_{\ll\in\P:\,\ell(\ll)\leq N}c_{\ll,N}^\pm s_\ll(\u),
\ee
where
\be
\label{coefficients}
c_{\ll,N}^\pm=\frac{\int_{(0,1)^N} \chi_\lambda(\underline x^{\pm 1}) \Delta^2(\underline x) \prod_{a=1}^Nx_a^\a(1-x_a)^\b\d^N\underline x}{\int_{(0,1)^N}\Delta^2(\underline x)\prod_{a=1}^Nx_a^\a(1-x_a)^\b\d^N\underline x}.
\ee

For the negative coefficients $c_{\ll,N}^-$ we shall use the following elementary lemma.

\begin{lemma}
\label{lemma:schurinverse}
For any partition $\ll=(\ll_1,\dots,\ll_\ell)$ of length $\ell\leq N$ we have
\be
\chi_\ll(\underline x^{-1})=\left(\prod_{a=1}^Nx_a^{-\ll_1}\right)\chi_{\wh\ll}(\underline x),
\ee
where $\wh\ll$ is the partition of length $<N$ whose parts are $\wh\ll_j=\ll_1-\ll_{N-j+1}$.
\end{lemma}

\begin{proof}
The proof follows from the following chain of equalities;
\begin{align}
\nonumber
\chi_\ll(\underline x^{-1})&=\frac{\det\left[x_i^{-N+j-\ll_j}\right]_{i,j=1}^N}{\det\left[x^{-N+j}_i\right]_{i,j=1}^N}
=\frac{\det\left[x_i^{1-j-\ll_{N-j+1}}\right]_{i,j=1}^N}{\det\left[x^{1-j}_i\right]_{i,j=1}^N}
\\
&
=\left(\prod_{a=1}^Nx_a^{-\ll_1}\right)\frac{\det\left[x_i^{N-j+\ll_1-\ll_{N-j+1}}\right]_{i,j=1}^N}{\det\left[x^{N-j}_i\right]_{i,j=1}^N}
=\left(\prod_{a=1}^Nx_a^{-\ll_1}\right)\chi_{\wh\ll}(\underline x).
\end{align}
In the first step we have shuffled the columns as $j\mapsto N-j+1$, then we have multiplied both numerator and denominator by $(x_1\cdots x_N)^{N+\ll_1}$, and finally we have applied the definition \eqref{characters}.
\end{proof}

For the simplification of the coefficients \eqref{coefficients} we rely on the following \emph{Schur--Selberg integral}
\be
\label{selberg}
\int_{(0,1)^N}\!\chi_\ll(\underline x)\Delta^2(\underline x)\prod_{a=1}^nx_a^\a(1-x_a)^\b\d^N\underline x=
 N! \!\prod_{1 \leq i < j \leq N}\!(\lambda_i-\lambda_j+j-i) \prod_{k=1}^N
\frac{\Gamma(\b+k)\Gamma(\a+N+\lambda_k-k+1)}
{\Gamma(\a+\b+2N+\lambda_k-k+1)},
\ee
for which we refer e.g.~to \cite[page 385]{Macdonald}.
The above allows us to prove the following proposition.

\begin{proposition}
We have
\be
\label{cpartitions}
c_{\ll,N}^+=\frac{\dim\ll}{|\ll|!}\prod_{(i,j) \in \lambda} \frac{(N-i+j) (\a+N-i+j)} {(\alpha+\beta+2N-i+j)}
,\
c_{\ll,N}^-=\frac{\dim\ll}{|\ll|!}\prod_{(i,j) \in \lambda} \frac{(N-i+j) (\a+\b+N+i-j)} {(\alpha+i-j)}.
\ee
\end{proposition}
\begin{proof}
We start with $c_{\ll,N}^+$; using \eqref{coefficients}, \eqref{selberg}, and \eqref{dimll} we compute
\begin{align}
\nonumber
c_{\ll,N}^+&= \frac{\prod_{1 \leq i < j \leq N}(\lambda_i-\lambda_j+j-i) }{\prod_{1 \leq i < j \leq N}(j-i)}
\prod_{k=1}^N\frac{\Gamma(\a+N+\lambda_k-k+1)\Gamma(\a+\b+2N-k+1)}{\Gamma(\a+\b+2N+\lambda_k-k+1)\Gamma(\a+N-k+1)}
\\
\nonumber
&=\frac{\dim\ll}{|\ll|!}
\prod_{k=1}^{N-1} \frac{(N-k+1)_{\lambda_k}(\a+N-k+1)_{\lambda_k}} {(\alpha+\beta+2N-k+1)_{\lambda_k}}
\\
&=\frac{\dim\ll}{|\ll|!}\prod_{(i,j) \in \lambda} \frac{(N-i+j) (\a+N-i+j)} {(\alpha+\beta+2N-i+j)}.
\end{align}
We remind that we are using the notation \eqref{risingfactorial} for the rising factorial.
For $c_{\ll,N}^-$ we first note that, thanks to Lemma \ref{lemma:schurinverse} and \eqref{selberg}, we have
\be
\int_{(0,1)^N}\chi_\ll(\underline x^{-1})\Delta^2(\underline x)\prod_{a=1}^Nx_a^\a(1-x_a)^\b\d^N\underline x=
N!\prod_{1 \leq i < j \leq N}(\lambda_i-\lambda_j+j-i) \prod_{k=1}^N\frac{\Gamma(\b+k)\Gamma(\a-\lambda_{k}+k)}{\Gamma(\alpha+\beta+N-\lambda_{k}+k)},
\ee
then with similar computations as above we obtain
\begin{align}
\nonumber
c_{\ll,N}^-&= \frac{\prod_{1 \leq i < j \leq N}(\lambda_i-\lambda_j+j-i) }{\prod_{1 \leq i < j \leq N}(j-i)}
\prod_{k=1}^N\frac{\Gamma(\a-\lambda_k+k)\Gamma(\a+\b+N+k)}{\Gamma(\a+\b+N-\lambda_k+k)\Gamma(\a+k)}
\\
\nonumber
&=\frac{\dim\ll}{|\ll|!}
\prod_{k=1}^{N-1} \frac{(N-k+1)_{\lambda_k}(\a+\b+N-\lambda_k+k)_{\lambda_k}} {(\alpha-\ll_k+k)_{\lambda_k}}
\\
&=\frac{\dim\ll}{|\ll|!}\prod_{(i,j) \in \lambda} \frac{(N-i+j) (\a+\b+N+i-j)} {(\alpha+i-j)}.
\end{align}
\end{proof}

This proposition enables us to identify the Jacobi generating function \eqref{Zpm} with the generating function of multiparametric weighted Hurwitz numbers in~\eqref{tauGmonomial} as follows.

\begin{corollary}
\label{cormain}
Let $c_\a:=1+\a/N$ and $c_\b:=1+\b/N$; then the Jacobi formal partition functions in \eqref{Zpm} take the form
\begin{align}
Z_N^+({\bf u})&=\tau_{G^+}\left(\epsilon=\frac 1N,\t\right),& G^+(z)&=\frac {(1+z)\left(1+\frac z{c_\a}\right)}{1+\frac z{c_\a+c_\b}},& t_k&=\left(\frac{c_\a N}{c_\a+c_\b}\right)^ku_k,
\\
Z_N^-({\bf u})&=\tau_{G^-}\left(\epsilon=\frac 1N,\t\right),& G^-(z)&=\frac {(1+z)\left(1-\frac z{c_\a+c_\b-1}\right)}{1-\frac z{c_\a-1}},& t_k&=\left(\frac{(c_\a+c_\b-1) N}{c_\a-1}\right)^ku_k,
\end{align}
where $\tau_G$ is introduced in Theorem \ref{prop}.
\end{corollary}

\begin{proof}
We first note that we can rewrite the expansion \eqref{schurexpansion} as
\be
Z_N^\pm(\u)=\sum_{\ll\in\P}c_{\ll,N}^\pm s_\ll(\u),
\ee
with the sum  over all partitions $\P$ and no longer restricted to $\ell(\ll)\leq N$; this is clear as $c_{N,\ll}^\pm=0$ whenever $N=0,1,2,\dots$ and $\ell(\ll)>N$.
Then the proof is immediate by the formula \eqref{rlambdaG} for the coefficients $r_\ll^{(G,\epsilon)}$, since \eqref{cpartitions} can be rewritten as
\begin{align}
c_{\ll,N}^+&=
\frac{\dim\ll}{|\ll|!}
\left(\frac{c_\a N}{c_\a+c_\b}\right)^{|\ll|}
\prod_{(i,j) \in \lambda}\frac{\left(1+\frac 1N(j-i)\right) \left(1+\frac 1{c_\a N}(j-i)\right)}{1+\frac 1{(c_\a+c_\b )N}(j-i)},
\\
c_{\ll,N}^-&=
\frac{\dim\ll}{|\ll|!}
\left(\frac{(c_\a+c_\b-1) N}{c_\a-1}\right)^{|\ll|}
\prod_{(i,j) \in \lambda} \frac{\left(1+\frac 1N(j-i)\right)\left(1-\frac 1{(c_\a+c_\b-1)N}(j-i)\right)}{1-\frac 1{(c_\a-1)N}(j-i)}.
\end{align}
\end{proof}

\subsection{Hurwitz numbers \texorpdfstring{$h_g(\ll,\mu,\nu)$}{h\_g} and multiparametric Hurwitz numbers}
We now connect the multiparametric Hurwitz numbers \eqref{eq:defhurwitz} for the functions $G^\pm(z)$, appearing in Corollary \ref{cormain}, with the counting problem in Definition \ref{def:hurwitz}.

\begin{proposition}\label{propinterpretationhurwitz}
If $G(z)=\frac{(1+z)(1+\gamma z)}{1-\delta z}$, with $\gamma$ and $\delta$ parameters, then for all partitions $\ll\vdash n$ and all integers $g\geq 0$ we have
\be
\label{eq:interpretationhurwitz}
H_G^{2g-2+n+\ell(\ll)}(\ll)=
\frac {1}{n!}
\sum_{\mu,\nu\vdash n}\gamma^{n-\ell(\nu)}\delta^{\ell(\mu)+\ell(\nu)+\ell(\ll)+2g-2-n}h_g(\ll,\mu,\nu),
\ee
where the triple monotone Hurwitz number $h_g(\ll,\mu,\nu)$ has been introduced in Definition \ref{def:hurwitz}.
\end{proposition}

\begin{proof}
We apply \eqref{JMvsC} to the first two factors of the following to get
\begin{align}
\nonumber
\prod_{a=1}^nG\left(\epsilon J_a\right)&=\prod_{a=1}^n(1+\epsilon\J_a)(1+\epsilon \gamma\J_a)\frac 1{1-\epsilon \delta\J_a}
\\
&=\left(\sum_{\mu\vdash n}\epsilon^{n-\ell(\mu)}\CC_\mu\right)\left(\sum_{\nu\vdash n}(\epsilon\gamma)^{n-\ell(\nu)}\CC_\nu\right)\left(\sum_{r\geq 0}(\epsilon\delta)^r\sum_{1\leq a_1\leq\dots\leq a_r\leq n}\J_{a_1}\cdots\J_{a_r}\right).
\end{align}
By definition \eqref{eq:defhurwitz}, extracting the coefficient of $\epsilon^d\CC_\ll$ and dividing by $z_\ll$ we obtain $H^d_G(\ll)$; therefore
\be
H_{G}^{d}(\ll)=
\frac {1}{z_\ll|\cyc(\ll)|}
\sum_{\mu,\nu\vdash n}\gamma^{n-\ell(\nu)}\delta^{r} h_g(\ll,\mu,\nu),
\ee
where $d,r,g$ in this identity are related via
\be
r=\ell(\ll)+\ell(\mu)+\ell(\nu)+2g-2-n,\qquad d=2n-\ell(\mu)-\ell(\nu)+r.
\ee
The proof is complete by the identity $z_\ll|\cyc(\ll)|=n!$.
\end{proof}

\subsection{Proof of Theorem \ref{thmhurwitz}}
From Corollary \ref{cormain} we have, with the scaling $\a=(c_\a-1)N$, $\b=(c_\b-1)N$,
\begin{align}
Z^+_N(\u)&=\sum_{d\geq 1}\frac 1{N^d}\sum_{\ll\in\P}
\left(\frac{c_\a N}{c_\a+c_\b}\right)^{|\ll|}
H^d_{G^+}(\ll)\prod_{i=1}^{\ell(\ll)}u_{\ll_i},
\\
Z^-_N(\u)&=\sum_{d\geq 1}\frac 1{N^d}\sum_{\ll\in\P}
\left(\frac{(c_\a+c_\b-1) N}{c_\a-1}\right)^{|\ll|}
H^d_{G^-}(\ll)\prod_{i=1}^{\ell(\ll)}u_{\ll_i},
\end{align}
where we have used Proposition \ref{prop}. It follows from \eqref{Zpm} that
\begin{align}
\frac{\left\langle\prod_{j=1}^\ell\tr X^{\ll_j}\right\rangle}{z_\ll}&=\sum_{d\geq 1}N^{|\ll|-d}\left(\frac{c_\a}{c_\a+c_\b}\right)^{|\ll|}H_{G^+}^d(\ll),
\\
\frac{\left\langle\prod_{j=1}^\ell\tr X^{-\ll_j}\right\rangle}{z_\ll}&=\sum_{d\geq 1}N^{|\ll|-d}\left(\frac{c_\a+c_\b-1}{c_\a-1}\right)^{|\ll|}H_{G^-}^d(\ll),
\end{align}
and using finally Proposition \ref{propinterpretationhurwitz} we have
\begin{align}
\frac{\left\langle\prod_{j=1}^\ell\tr X^{\ll_j}\right\rangle}{z_\ll}&=
\frac 1{|\ll|!}\sum_{g\geq 0}N^{2-2g-\ell(\ll)}\sum_{\mu,\nu\vdash|\ll|}(-1)^{|\ll|}
\frac{c_\a^{\ell(\nu)}}{(-c_\a-c_\b)^{\ell(\mu)+\ell(\nu)+\ell(\ll)+2g-2}}h_g(\ll,\mu,\nu),
\\
\frac{\left\langle\prod_{j=1}^\ell\tr X^{-\ll_j}\right\rangle}{z_\ll}&=
\frac 1{|\ll|!}\sum_{g\geq 0}N^{2-2g-\ell(\ll)}\sum_{\mu,\nu\vdash n}(-1)^{|\ll|}
\frac {\left(1-c_\a-c_\b\right)^{\ell(\nu)}}{\left(c_\a-1\right)^{\ell(\mu)+\ell(\nu)+\ell(\ll)+2g-2}}h_g(\ll,\mu,\nu).
\end{align}
The proof is complete.\hfill$\square$

\begin{remark}
\label{remLUE2}
Let us note that letting $c_\b\to\infty$ in the functions $G^\pm$ of Corollary \ref{cormain} we have $G^+(z)\to (1+z)(1+z/c_\a)$ and $G^-(z)\to (1+z)/(1-z/(c_\a-1))$. The Hurwitz numbers corresponding to these limit functions can be identified as in Proposition \ref{propinterpretationhurwitz} in terms of double strictly ($+$) or weakly ($-$) Hurwitz numbers, respectively. Thus, bearing in mind the scaling limit for $\b\to\infty$ of JUE correlators to the correlators of the Laguerre Unitary Ensemble of Remark \ref{LUElimit}, the Theorem \ref{thmhurwitz} recovers the results of \cite{CDO2018}.
\end{remark}

\section{Computing correlators of Hermitian models}\label{sec3}

In this section we prove\footnote{The proof presented in the published version of this preprint is based on a different strategy but it contains an error. Here we give a different proof.} Theorem~\ref{thmcumulants}. 
The strategy is based on the observation that setting
\be
\mathscr Z_N(t,z):=\int_{\H_N(I)}\exp\bigl(\tr(V(X)+\sum_{i=1}^\ell t_i(z_i-X)^{-1})\bigr)\d X,
\quad t=(t_1,\dots,t_\ell),\ z=(z_1,\dots,z_\ell),
\ee
we have
\be
\label{key}
\mathscr C^{\sf c}(z_1,\dots,z_\ell)=\frac{\pa^\ell}{\pa t_1\cdots\pa t_\ell}\log\mathscr Z(t,z)\biggr|_{t_i=0}.
\ee
Here and below it is assumed that $z_i\not\in I$.

\subsection{Orthogonal polynomials on the real line and unitary-invariant ensembles}

We denote $P_\ell(z)$ the monic orthogonal polynomials, $h_\ell=\int_IP^2_\ell(x)\e^{V(x)}\d x$, see \eqref{monicOPintro}, and
\be
\label{CauchyTransform}
\wh P_\ell(z) := \frac{1}{2 \pi \i}\int_IP_\ell(x)\e^{V(x)}\frac{\d x}{x-z}
\ee
their \emph{Cauchy transforms}. The matrix 
\be
\label{YN}
Y_N(z):=
\renewcommand{\arraystretch}{1.25}\left(\begin{array}{cc}
P_N(z) & \wh P_N(z)
\\
-\frac{2\pi\i}{h_{N-1}} P_{N-1}(z) & -\frac{2\pi\i}{h_{N-1}}\wh P_{N-1}(z)
\end{array}\right),
\ee
introduced in \eqref{Ymatrix}, is an analytic function of $z\in\C\setminus I$. It satisfies the \emph{jump condition}
\be
\label{jumpY}
Y_{N,+}(x)=Y_{N,-}(x)\begin{pmatrix}
1 & \e^{V(x)} \\ 0 & 1 
\end{pmatrix},
\qquad x\in I^\circ,
\ee
where we use the notation
\be
Y_{N,\pm}(x)=\lim_{\epsilon\to 0_+}Y_N(x\pm\i\epsilon),\qquad x\in I^\circ,
\ee
and $I^\circ$ is the interior of the interval $I$. As $z\to\infty$ we have
\be
\label{growthYinfinity}
Y_{N}(z)=\left(\1+\O(z^{-1})\right)z^{N\s_3},
\ee
where we denote $\mathbf 1=\begin{pmatrix}
1& 0 \\ 0 & 1
\end{pmatrix}$ and
$\s_3=\begin{pmatrix}
1& 0 \\ 0 & -1
\end{pmatrix}$. Lastly, we recall the Christoffel--Darboux identity
\be
\label{CDdef}
K_N(x,y):=\e^{\frac{V(x)+V(y)}2}\sum_{i=0}^{N-1}\frac{P_i(X)P_i(y)}{h_i}=\frac{\e^{\frac{V(x)+V(y)}2}}{h_{N-1}}\frac{P_N(x)P_{N-1}(y)-P_{N-1}(x)P_{N}(y)}{x-y}\,,
\ee
expressing the kernel (\emph{Christoffel--Darboux kernel}) of the orthogonal projector onto the space of polynomials of degree $<N$ in the Hilbert space $L^2(I,\e^{V(x)}\d x)$.
Therefore, the Christoffel--Darboux kernel can be conveniently rewritten in terms of the matrix $Y_N(z)$ in \eqref{YN} as
\begin{equation}
\label{CD}
K_N(x,y)=
-\frac{\e^{\frac{V(x)+V(y)}2}}{2\pi\i(x-y)}\begin{pmatrix}0 & 1 \end{pmatrix}Y^{-1}_N(x)Y_N(y)\begin{pmatrix}1 \\ 0\end{pmatrix},
\end{equation}
which is independent of the choice of boundary value of $Y_N$ because of \eqref{jumpY}.

Next, we need to recall the connection of orthogonal polynomials to the theory of unitary-invariant ensembles of random matrices.
The main point which is relevant for our present purposes is that~\cite{Deift}
\be
\int_{\H_N(I)}\exp\bigl(\tr(V(X))\bigr)\d X = N!h_0^V\cdots h_{N-1}^V\,,
\ee
where it is convenient to explicitly express the dependence of $P_\ell=P_\ell^V$ and $h_\ell=h_\ell^V$ on the potential $V$.
Therefore, introducing the modified potential
\be
\label{Vtz}
V_{t,z}(x):=V(x)+\sum_{i=1}^\ell \frac{t_i}{z_i-x},\qquad t=(t_1,\dots,t_\ell),\ \ z=(z_1,\dots,z_\ell),
\ee
we have
\be
\label{firstrel}
\mathscr Z_N(t,z)=N!h_0^{V_{t,z}}\cdots h_{N-1}^{V_{t,z}}.
\ee
\begin{lemma}
\label{lemmasecondrel}
We have
\be
\label{secondrel}
\pa_{t_j}h_i^{V_{t,z}}=\int_{I}\bigl(P_i^{V_{t,z}}(x)\bigr)^2\e^{V_{t,z}(x)}\frac{\d x}{z_j-x}.
\ee
\end{lemma}
\begin{proof}
We have $h_i^{V_{t,z}}=\int_{I}\bigl(P_i^{V_{t,z}}(x)\bigr)^2\e^{V_{t,z}(x)}\d x$ hence
\be
\pa_{t_j}h_i^{V_{t,z}}=2\int_{I}P_i^{V_{t,z}}(x)\left(\pa_{t_j}P_i^{V_{t,z}}(x)\right)\e^{V_{t,z}(x)}\d x+\int_{I}\bigl(P_i^{V_{t,z}}(x)\bigr)^2\e^{V_{t,z}(x)}\left(\pa_{t_j}V_{t,z}(x)\right)\d x\,,
\ee
but the first term vanishes by orthogonality because $P_i^{V_{t,z}}(x)$ are normalized to be monic and, therefore, $\pa_{t_j}P_i^{V_{t,z}}(x)$ is a polynomial of degree strictly less than~$i$.
\end{proof}

\subsection{Case \texorpdfstring{$\ell=1$}{l=1}}

It follows from \eqref{CD} that 
\be
\label{eqrho1}
K_N(x,x)=\lim_{y\to x}K_N(x,y)=\frac{\e^{V(x)}}{2\pi\i}\begin{pmatrix}0 & 1 \end{pmatrix}Y^{-1}_N(x)Y_N'(x)\begin{pmatrix}1 \\ 0\end{pmatrix}.
\ee
In the following we shall use the notation 
\be
\Delta f(x):=f_+(x)-f_-(x),\qquad x\in I^\circ,
\ee
for the jump of a function~$f$ across~$I$, namely $f_\pm(x):=\lim_{\epsilon\to 0_+}f(x\pm\i\epsilon)$.
The next lemma is well known, see e.g.~\cite{CGM2015}, and it is proven here for the reader's convenience.

\begin{lemma}
\label{lemmajumponepoint}
We have
\be
K_N(x,x)=-\frac 1{2\pi\i}\Delta\left[\tr\left(Y^{-1}_N(x)\frac{\pa Y_N(x)}{\pa x}\mathrm{E}_{1,1}\right)\right],\qquad \mathrm{E}_{1,1}:=\begin{pmatrix}
1 & 0 \\ 0 & 0
\end{pmatrix}
\ee
\end{lemma}

\begin{proof}
Let us denote $':=\pa_x$.
It follows from the jump condition \eqref{jumpY} for $Y_N$ that
\be
Y'_{N,+}(x)=Y'_{N,-}(x)\begin{pmatrix}
1 & \e^{V(x)} \\ 0 & 1 
\end{pmatrix}+Y_{N,-}(x)\begin{pmatrix}
0 & V'(x)\e^{V(x)} \\ 0 & 0
\end{pmatrix},
\qquad x\in I^\circ.
\ee
Therefore we compute
\begin{align}
\nonumber
&\Delta\left[\tr\left(Y^{-1}_N(x)Y_N'(x)\mathrm{E}_{1,1}\right)\right]=\tr\left(Y^{-1}_{N,+}(x)Y_{N,+}'(x)\mathrm{E}_{1,1}\right)-\tr\left(Y^{-1}_{N,-}(x)Y_{N,-}'(x)\mathrm{E}_{1,1}\right)
\\
\nonumber
&\qquad
=\tr\left(
\begin{pmatrix}
1 & -\e^{V(x)} \\ 0 & 1
\end{pmatrix}
Y^{-1}_{N,-}(x)Y_{N,-}'(x)
\begin{pmatrix}
1 & \e^{V(x)} \\ 0 & 1
\end{pmatrix}\mathrm{E}_{1,1}\right)-
\tr\left(Y^{-1}_{N,-}(x)Y_{N,-}'(x)\mathrm{E}_{1,1}\right)\\
&\qquad\qquad\qquad\qquad\qquad\qquad\qquad\qquad\qquad
+
\tr\left(
\begin{pmatrix}
1 & -\e^{V(x)} \\ 0 & 1
\end{pmatrix}
\begin{pmatrix}
0 & V'(x)\e^{V(x)} \\ 0 & 0
\end{pmatrix}\mathrm{E}_{1,1}\right)
\end{align}
The last term vanishes and so, by the cyclic property of the trace, we have
\be
\Delta\left[\tr\left(Y^{-1}_N(x)Y_N'(x)\mathrm{E}_{1,1}\right)\right]=\tr\left[Y^{-1}_{N,-}(x)Y_{N,-}'(x)\left(\begin{pmatrix}
1 & \e^{V(x)} \\ 0 & 1
\end{pmatrix}\mathrm{E}_{1,1}\begin{pmatrix}
1 & -\e^{V(x)} \\ 0 & 1
\end{pmatrix}-\mathrm{E}_{1,1}\right)\right]
\ee
which is easily seen to be equivalent, up to multiplying by~$-1/(2\pi\i)$, to~\eqref{eqrho1}.
\end{proof}

We are ready for the proof of the case $\ell=1$. In such case, $t=t_1,z=z_1$ and $V_{t,z}(x)=V(x)+t/(z-x)$.
By~\eqref{firstrel}, Lemma~\ref{lemmasecondrel}, and~\eqref{CDdef}, we have
\be
\pa_{t}\log\mathscr Z_N(t,z)=\sum_{i=0}^{N-1}\frac 1{h_i^{V_{t,z}}}\pa_{t} h_{i}^{V_{t,z}}
=\sum_{i=0}^{N-1}\frac 1{h_i^{V_{t,z}}}\int_{I}\bigl(P_i^{V_{t,z}}(x)\bigr)^2\e^{V_{t,z}(x)}\frac{\d x}{z-x}
=\int_I K_N^{V_{t,z}}(x,x)\frac {\d x}{z-x},
\ee
where we denote explicitly the dependence of the Christoffel--Darboux kernel on the potential.
Let $\G$ be an oriented contour in the complex plane which surrounds $I$ in counterclockwise sense (i.e., $I$ lies on the left of $\G$) and leaves $z$ outside (i.e., $z$ lies to the right of $\G$).
Then, using Lemma~\ref{lemmajumponepoint} we get
\begin{align}
\nonumber
\pa_{t}\log\mathscr Z_N(t,z)&=-\int_I \Delta\left[\tr\left(Y^{-1}_N(x;t,z)\frac{\pa Y_N(x;t,z)}{\pa x}\mathrm{E}_{1,1}\right)\right]\frac{\d x}{2\pi\i(z-x)}
\\
&= \int_\Gamma \tr\left(Y^{-1}_N(x;t,z)\frac{\pa Y_N(x;t,z)}{\pa x}\mathrm{E}_{1,1}\right)\frac{\d x}{2\pi\i(z-x)},
\end{align}
where $Y_N(\cdot;t,z)$ is the matrix~\eqref{YN} for the potential $V_{t,z}$. The last contour integral can be evaluated by a residue computation as
\be
\pa_{t}\log\mathscr Z_N(t,z)=\left(-\res{x=z}-\res{x=\infty}\right)\tr\left(Y^{-1}_N(x;t,z)\frac{\pa Y_N(x;t,z)}{\pa x}\mathrm{E}_{1,1}\right)\frac{\d x}{z-x}.
\ee
It can be checked from~\eqref{growthYinfinity} that the residue at $x=\infty$ vanishes. Therefore
\be
\label{eq:finalell1}
\pa_{t}\log\mathscr Z_N(t,z)=\tr\left(Y^{-1}_N(z;t,z)\left.\frac{\pa Y_N(x;t,z)}{\pa x}\right|_{x=z}\mathrm{E}_{1,1}\right).
\ee
Evaluating this identity at $t=0$, taking into account~\eqref{key}, we obtain exactly~\eqref{OnePointThm}.

\subsection{Case \texorpdfstring{$\ell=2$}{l=2}}

Let us first formulate a result that will be needed for all $\ell\geq 2$.
Let
\be
R(x;z,t):=Y_N(x;t,z)\mathrm{E}_{1,1}Y_N^{-1}(x;t,z).
\ee

\begin{lemma}
Let $V_{t,z}(x)=V(x)+\sum_{i=1}^\ell\tfrac{t_i}{z_i-x}$, $t=(t_1,\dots,t_\ell)$, and $z=(z_1,\dots,z_\ell)$.
For all $1\leq j\leq\ell$, we have
\be
\label{eqlemmanojump}
\frac 1{z_j-x}R(x;t,z)+\left(\frac{\pa Y_N}{\pa t_j}(x;t,z)\right)Y_N^{-1}(x;t,z)=\frac 1{z_j-x}R(z_j;t,z).
\ee
\end{lemma}
\begin{proof}
Let us denote by $\Omega_j(x;t,z)$ the left-hand side of~\eqref{eqlemmanojump}.
Using~\eqref{jumpY} we get the identities
\begin{align}
Y_N(x_+;t,z)&=Y_N(x_-;t,z)\begin{pmatrix}
1 & \e^{V_{t,z}(x)} \\ 0 & 1
\end{pmatrix}\,,\\
\frac{\pa Y_N}{\pa t_j}(x_+;t,z)&=\frac{\pa Y_N}{\pa t_j}(x_-;t,z)\begin{pmatrix}
1 & \e^{V_{t,z}(x)} \\ 0 & 1
\end{pmatrix}+Y_N(x_-;t,z)\begin{pmatrix}
0 & \frac 1{z_j-x}\e^{V_{t,z}(x)} \\ 0 & 0
\end{pmatrix}\,,
\end{align}
from which we readily ascertain that $\Delta\Omega_j(x;t,z)=0$ for all $x\in\R$.
Hence, $\Omega_j(x;t,z)$ is a meromorphic function of $x$ with a single simple pole at $x=z_j$ and which vanishes at $x=\infty$, because of~\eqref{growthYinfinity}, and so the statement follows.
\end{proof}

Let us consider the case $\ell=2$, in which $t=(t_1,t_2)$, $z=(z_1,z_2)$, and $V_{t,z}(x)=V(x)+\tfrac{t_1}{z_1-x}+\tfrac{t_2}{z_2-x}$.
By the argument used for $\ell=1$, cf.~\eqref{eq:finalell1}, we obtain
\be
\pa_{t_1}\log\mathscr Z_N(t,z)=\tr\left(Y^{-1}_N(z_1;t,z)\left.\frac{\pa Y_N(x;t,z)}{\pa x}\right|_{x=z_1}\mathrm{E}_{1,1}\right).
\ee
Next we have to take a derivative in $t_2$: omitting the explicit dependence on $t,z$, we have
\be
\pa_{t_2}\pa_{t_1}\log\mathscr Z_N(t,z)=
\tr\biggl(-Y^{-1}_N(z_1)\frac{\pa Y_N}{\pa t_2}(z_1)Y_N^{-1}(z_1)\left.\frac{\pa Y_N(x)}{\pa x}\right|_{x=z_1}\mathrm{E}_{1,1}
+Y^{-1}_N(z_1)\left.\frac{\pa^2 Y_N(x)}{\pa t_2\pa x}\right|_{x=z_1}\mathrm{E}_{1,1}\biggr).
\ee
We use~\eqref{eqlemmanojump} to rewrite the first term inside the trace in the right-hand side as
\be
\label{eq:duepunti1}
-Y^{-1}_N(z_1)\frac{R(z_2)-R(z_1)}{z_2-z_1}\left.\frac{\pa Y_N(x)}{\pa x}\right|_{x=z_1}\mathrm{E}_{1,1}
\ee
and the second term as
\begin{align}
\nonumber
Y^{-1}_N(z_1)\frac{\pa^2 Y_N(x)}{\pa x\pa t_2}\bigg|_{x=z_1}\mathrm{E}_{1,1}
&=
Y^{-1}_N(z_1)\frac{\pa}{\pa x}\biggl(\frac{R(z_2)-R(x)}{z_2-x}Y_N(x)\biggr)\bigg|_{x=z_1}\mathrm{E}_{1,1}
\\
\nonumber
&=Y^{-1}_N(z_1)\biggl(
\frac{R(z_2)-R(z_1)}{(z_2-z_1)^2}Y_N(z_1)
-\frac{\biggl[\frac{\pa Y_N(x)}{\pa x}\bigg|_{x=z_1}Y^{-1}_N(z_1),R(z_1)\biggr]}{z_2-z_1}Y_N(z_1)
\\
\label{eq:duepunti2}
\nonumber
&\qquad\qquad\qquad\qquad\qquad
+\frac{R(z_2)-R(z_1)}{z_2-z_1}\frac{\pa Y_N(x)}{\pa x}\bigg|_{x=z_1}\biggr)\mathrm{E}_{1,1},
\end{align}
where $[A,B]:=AB-BA$ is the commutator.
The term in the last row exactly cancels with~\eqref{eq:duepunti1}, and so, rearranging terms
\be
\label{toberewritten}
\pa_{t_2}\pa_{t_1}\log\mathscr Z_N(t,z)=\frac{\tr\bigl(R(z_1)R(z_2)\bigr)-1}{(z_2-z_1)^2}+
\frac{\tr\left(\biggl[Y^{-1}_N(z_1)\frac{\pa Y_N(x)}{\pa x}\bigg|_{x=z_1},\mathrm{E}_{1,1}\biggr]\mathrm{E}_{1,1}\right)}{z_2-z_1}
\ee
and, since $\tr([A,B]B)=\tr([AB,B])=0$, the proof of the case $\ell=2$ is completed by setting $t_1=t_2=0$.

\subsection{Case \texorpdfstring{$\ell\geq 3$}{l>2}}

Let us denote
\be
S_\ell(z_1,\dots,z_\ell;t):=-\sum_{(i_1,\dots,i_\ell)\in\cyc((\ell))}\frac{\tr\left(R(z_{i_1};t,z)\cdots R(z_{i_\ell};t,z)\right)}{(z_{i_1}-z_{i_2})\cdots(z_{i_{\ell-1}}-z_{i_\ell})(z_{i_\ell}-z_{i_1})}-\frac{\delta_{\ell,2}}{(z_1-z_2)^2},
\ee
where the sum extends over cyclic permutations of $\{1,\dots,\ell\}$.
We aim at proving that
\be
\label{tobeprovedbyinduction}
\frac{\pa^\ell \log\mathscr Z_N(t,z)}{\pa t_{\ell}\cdots\pa t_{1}}=S_\ell(z_1,\dots,z_\ell;t).
\ee
where $Y_N(x;t,z)$, and so $R(x;t,z)$, are computed for the potential $V_{t,z}(x)=V(x)+\sum_{i=1}^\ell\frac{t_i}{z_i-x}$. Then, \eqref{MultiPointThm} follows by taking~$t_i=0$.
The proof of~\eqref{tobeprovedbyinduction} is  by induction on~$\ell\geq 2$ and it is similar in spirit to that in~\cite{BDY2016,BDY2018,GGR2020}.

Let us assume~\eqref{tobeprovedbyinduction} for $\ell$ and let us prove it for $\ell+1$. Since the potential $V$ is arbitrary, we can assume~\eqref{tobeprovedbyinduction} holds true for $V_{t,z}(x)=V(x)+\sum_{j=1}^{\ell+1}\frac{t_{j}}{z_{j}-x}$, and so we just have to show that $\pa_{t_{\ell+1}}S_\ell(z_1,\dots,z_\ell;t)$ is equal to $S_{\ell+1}(z_1,\dots,z_\ell,z_{\ell+1};t)$.
To this end we first observe that by~\eqref{eqlemmanojump}, we have
\be
\label{eq:dtR}
\frac{\pa R(x;t,z)}{\pa t_j}=\biggl[\frac{R(z_j;t,z)-R(x;t,z)}{z_j-x},R(x;t,z)\biggr]=\frac{[R(z_j;t,z),R(x;t,z)]}{z_j-x}.
\ee
Therefore,
\be
\frac{\pa S_\ell(z_1,\dots,z_\ell;t)}{\pa t_{\ell+1}}=-\sum_{(i_1,\dots,i_\ell)\in\cyc((\ell))}\sum_{j=1}^\ell\frac{\tr\left(R(z_{i_1};t,z)\cdots[R(z_{\ell+1};t,z),R(z_{i_j};t,z)]\cdots R(z_{i_\ell};t,z)\right)}{(z_{i_1}-z_{i_2})\cdots(z_{i_\ell}-z_{i_1})(z_{\ell+1}-z_{i_j})}
\ee
Expanding the commutator $[R(z_{\ell+1};t,z),R(z_{i_j};t,z)]=R(z_{\ell+1};t,z)R(z_{i_j};t,z)-R(z_{i_j};t,z)R(z_{\ell+1};t,z)$, we note that in the previous sum, each term involving the expression
\be
\tr\left(R(z_{i_1};t,z)\cdots R(z_{\ell+1};t,z)R(z_{i_j};t,z)\cdots R(z_{i_\ell};t,z)\right)
\ee
appears twice, but with different denominators.
Collecting such terms yields
\begin{align}
\nonumber
&\sum_{(i_1,\dots,i_\ell)\in\cyc((\ell))}\sum_{j=1}^\ell\frac{\tr\left(R(z_{i_1};t,z)\cdots R(z_{\ell+1};t,z)R(z_{i_j};t,z)\cdots R(z_{i_\ell};t,z)\right)}{(z_{i_1}-z_{i_2})\cdots(z_{i_\ell}-z_{i_1})}\biggl(\frac 1{z_{i_j}-z_{\ell+1}}-\frac 1{z_{i_{j-1}}-z_{\ell+1}}\biggr)
\\
\nonumber
&\qquad=-\sum_{(i_1,\dots,i_\ell)\in\cyc((\ell))}\sum_{j=1}^\ell\frac{\tr\left(R(z_{i_1};t,z)\cdots R(z_{\ell+1};t,z)R(z_{i_j};t,z)\cdots R(z_{i_\ell};t,z)\right)}{(z_{i_1}-z_{i_2})\cdots(z_{i_{j-1}}-z_{\ell+1})(z_{\ell+1}-z_{i_j}) \cdots (z_{i_\ell}-z_{i_1})}
\\
&\qquad=S_{\ell+1}(z_1,\dots,z_\ell,z_{\ell+1}),
\end{align}
where we set $i_0:=i_\ell$ in the internal summation. The proof is complete.

\begin{remark}
\label{remarkscalarseries}
We note here that since $R(z)$ is a rank one matrix, the formulae of Theorem \ref{thmcumulants} for $\mathscr C^{\sf c}_\ell$, $\ell\geq 2$, can be expressed in terms of the scalar quantities
\be
w(x,y):=\frac{2\pi\i}{h_{N-1}}\frac{\pi_N(x)\wh\pi_{N-1}(y)-\pi_{N-1}(x)\wh\pi_{N}(y)}{x-y}
\ee
as
\be
\mathscr C_\ell^{\sf c}(z_1,\dots,z_\ell)=-\sum_{(i_1,\dots,i_\ell)\in\cyc( (\ell) )}w(z_{i_1},z_{i_2})\cdots w(z_{i_{\ell-1}},z_{i_\ell})w(z_{i_\ell},z_{i_1})-\frac{\delta_{\ell,2}}{(z_1-z_2)^2},\quad\ell\geq 2,
\ee
compare for instance with \cite{DYZ2020,Y2020}.
\end{remark}

\section{JUE correlators and Wilson Polynomials}\label{sec4}

In this section we prove Corollary \ref{corollaryWilson}. This is done by expanding the general formul\ae\ of Theorem \ref{thmcumulants} as $z_i\to 0,\infty$.
To this end we consider the monic orthogonal polynomials for the Jacobi measure, which are the classical (monic) Jacobi polynomials
\be
\label{monicJacobi}
P^{\sf J}_\ell(z) = 
\frac{\ell!}{(\a+\beta+\ell+1)_\ell}\sum_{k=0}^{\ell} \binom{\ell+\a}k\binom{\ell+\beta}  {\ell-k}\left( z-1 \right )^k z^{\ell-k},
\ee
satisfying the orthogonality property
\be \label{NormingC}
\int_0^1 P^{\sf J}_\ell(x)P^{\sf J}_m(x)x^\a(1-x)^\b\d x=h_\ell^{\sf J}\delta_{\ell,m},\qquad 
h_\ell^{\sf J}=\frac{\ell!\,\Gamma (\a+\ell+1) \Gamma (\beta+\ell+1) \Gamma (\a+\beta+\ell+1)}{\G(\a+\beta+2\ell+1) \Gamma (\a+\beta+2\ell+2)}.
\ee

\subsection{Expansion of the matrix \texorpdfstring{$R$}{R}}

This paragraph is devoted to the proof of the following proposition.

\begin{proposition} \label{Rexpansions}
We have the Taylor expansion at $z=\infty$
\be
R(z)=T^{-1}R^{[\infty]}(z)T,\qquad |z|>1,
\ee
where $T$ is the constant matrix \eqref{Tmatrix} and $R^{[\infty]}(z)$ is the matrix-valued power series in $z^{-1}$ in \eqref{Rinftyintro}.
We have the Poincar\'e asymptotic expansion at $z=0$ uniformly within the sector $0<\arg z<2\pi$
\be
R(z)\sim T^{-1}R^{[0]}(z)T,
\ee
where $T$ is the constant matrix \eqref{Tmatrix} and $R^{[0]}(z)$ is the matrix-valued (formal) power series in $z$ in \eqref{Rinftyintro}.
\end{proposition}

Looking back at the definition of the matrix $R(z)$, 
\be 
R(z):=Y_N(z)\left(\begin{array}{cc}1 & 0 \\ 0 & 0 \end{array}\right)Y_N^{-1}(z) =
\renewcommand*{\arraystretch}{1.2}\left(\begin{array}{cc}
-\frac{2 \pi \i}{h_{N-1}} P^{\sf J}_N(z) \widehat{P}^{\sf J}_{N-1}(z) & - P^{\sf J}_N(z) \widehat{P}^{\sf J}_N(z)
\\
-\left( \frac{2 \pi \i}{h_{N-1}} \right)^2 P^{\sf J}_{N-1}(z) \widehat{P}^{\sf J}_{N-1}(z) & -\frac{2 \pi \i}{h_{N-1}} P^{\sf J}_{N-1}(z) \widehat{P}^{\sf J}_N(z)
\end{array}\right),
\ee
we notice that it is sufficient to compute the expansions of the product of the Jacobi polynomials with their Cauchy transforms at the prescribed points. To this end, recall the explicit formula \eqref{monicJacobi} for the monic Jacobi orthogonal polynomials,
which can be rewritten as the \emph{Rodrigues' formula}
\be \label{rodrigues}
P^{\sf J}_\ell(z) =\frac{(-1)^\ell }{(\a+\beta+\ell+1)_\ell} z^{-\a} (1-z)^{-\beta} \frac{\d^\ell}{\d z^\ell}\left [z^{\a+\ell} (1-z)^{\beta+\ell}  \right].
\ee
The Cauchy transforms $\wh P_\ell^{\sf J}(z)$ defined in \eqref{CauchyTransform} can be expanded as stated below.

\begin{lemma}\label{CauchyExp}
The following relations hold true:
\begin{align}
\widehat{P}^{\sf J}_\ell(z) 	={}&-\frac{1}{2\pi\i (\a+\beta+\ell+1)_\ell}
 \sum_{j\geq 0}\frac {1}{z^{j+\ell+1}}(j+1)_\ell \frac{\G(\a+\ell+j+1)\G(\b+\ell+1)}{\G(\a+\b+2\ell+j+1)},\quad |z|>1,	\label{CauchyInf}		\\
\widehat{P}^{\sf J}_\ell(z) 	\overset{z \to 0}{\sim}&\, (-1)^\ell\frac{1}{2\pi\i (\a+\beta+\ell+1)_\ell}
 \sum_{j\geq 0} z^j (j+1)_\ell  \frac{\G(\a-j)\G(\b+\ell+1)}{\G(\a+\b+\ell-j+1)},				\label{Cauchy0}
\end{align}
where the first relation is a genuine Taylor expansion at $z=\infty$ whilst the second one is a Poincar\'e asymptotic expansion at $z=0$ uniform in the sector $0<\arg z<2\pi$.
\end{lemma}
\begin{proof}
We start with the expansion \eqref{CauchyInf} at $z=\infty$, which is computed as follows;
\begin{align}
\nonumber
\widehat{P^{\sf J}_\ell}(z)&=\frac{1}{2\pi\i}\int_0^{1} P^{\sf J}_\ell(x)x^{\a}(1-x)^{\beta}\frac{\d x}{x-z}
\\
\nonumber
&\overset{(i)}{=}-\frac{1}{2\pi\i }\sum_{j\geq 0}\frac 1{z^{j+1}}\int_0^{1}P^{\sf J}_\ell (x)x^{\a+j}(1-x)^{\beta}\d x
\\
\nonumber
&\overset{(ii)}{=}-\frac{1}{2\pi\i } \sum_{j\geq 0}\frac 1{z^{j+\ell+1}}\int_0^{1}P^{\sf J}_\ell (x)x^{\a+j+\ell}(1-x)^{\beta}\d x
\\
\nonumber
&\overset{(iii)}{=}-\frac{1}{2\pi\i } \frac{(-1)^\ell }{(\a+\beta+\ell+1)_\ell} \sum_{j\geq 0}\frac 1{z^{j+\ell+1}}\int_0^{1} \left(\frac{\d^\ell}{\d x^\ell}x^{\a+\ell} (1-x)^{\beta+\ell}\right)x^{j+\ell}\d x
\\
\nonumber
&\overset{(iv)}{=}-\frac{1}{2\pi\i } \frac{1}{(\a+\beta+\ell+1)_\ell} \sum_{j\geq 0}\frac 1{z^{j+\ell+1}}\int_0^{1}x^{\a+\ell}  (1-x)^{\beta+\ell} \frac{\d^\ell}{\d x^\ell}(x^{j+\ell})\d x
\\
\nonumber
&\overset{(v)}{=}-\frac{1}{2\pi\i } \frac{1}{(\a+\beta+\ell+1)_\ell} \sum_{j\geq 0}\frac {(j+1)_\ell}{z^{j+\ell+1}}\int_0^{1} x^{\a+\ell+j} (1-x)^{\beta+\ell} \d x
\\
&\overset{(vi)}{=}-\frac{1}{2\pi\i } \frac{1}{(\a+\beta+\ell+1)_\ell}
 \sum_{j\geq 0}\frac {1}{z^{j+\ell+1}}(j+1)_\ell \frac{\G(\a+\ell+j+1)\G(\b+\ell+1)}{\G(\a+\b+2\ell+j+1)}.
\end{align}
In $(i)$ we have expanded the geometric series and exchanged sum and integral by Fubini theorem, in $(ii)$ we use that $P^{\sf J}_\ell(z)$ is orthogonal to $z^j$ for $j<\ell$, in $(iii)$ we use the Rodrigues' formula \eqref{rodrigues}, in $(iv)$ we integrate by parts, in $(v)$ we compute the derivative, and finally in $(vi)$ we use the Euler beta integral.
The computation at $z=0$ is completely analogous, with the only difference that in $(i)$ it is not legitimate to exchange sum and integral so this step holds only in the sense of a Poincar\'e asymptotic series.
\end{proof}

The next step is to compute the expansions of the products of the Jacobi polynomials and their Cauchy transforms. To this end it is convenient to study more in detail the properties of $R(z)$. 

\begin{proposition} \label{DiffEqsProp}
The matrix $\Psi_N(z):=Y_N(z)z^{\a\s_3/2}(1-z)^{\b\s_3/2}$ satisfies the following linear differential equation
\be \label{Ydiffeq}
\pa_z\Psi_N(z) = U(z)\Psi_N(z)
\ee
and the matrix $R(z)$ satisfies the following Lax differential equation,
\be 	\label{LaxR}
\pa_z R(z)=[U(z),R(z)].
\ee
Here the matrix $U(z)$ is explicitly given as
\be 
\label{eq:U}
U(z)=\frac{U_0}{z}+\frac{U_1}{1-z},
\ee
with
\begin{align} 
U_0 &=
\begin{pmatrix}
 \frac{2 N (\a+\b+N)+\a (\a+\b)}{2 (\a+\b+2N)} & -\frac{h_N^{\sf J}}{2 \pi\i} (\a+\b+2N+1) \\
 \frac{2 \pi\i}{h_{N-1}^{\sf J}} (\a+\b+2N-1) & -\frac{2 N (\a+\b+N)+\a (\a+\b)}{2 (\a+\b+2N)}
\end{pmatrix},	\\ \label{Amatrix}
U_1 &=
\begin{pmatrix}
 -\frac{2 N (\a+\b+N)+\b (\a+\b)}{2 (\a+\b+2N)} &-\frac{h_N^{\sf J}}{2 \pi\i} (\a+\b+2N+1) \\
 \frac{2 \pi\i}{h_{N-1}^{\sf J}} (\a+\b+2N-1)  & \frac{2 N (\a+\b+N)+\b (\a+\b)}{2 (\a+\b+2N)}
\end{pmatrix}. 
\end{align}
\end{proposition}
\begin{proof}
From the definition $R(z):=Y_N(z)\begin{pmatrix}1 & 0 \\ 0 & 0 \end{pmatrix}Y_N^{-1}(z)$, we obtain $R(z)=\Psi_N(z)\begin{pmatrix}1 & 0 \\ 0 & 0 \end{pmatrix}\Psi_N^{-1}(z)$; therefore the Lax equation \eqref{LaxR} follows from \eqref{Ydiffeq}. The latter is a classical property of Jacobi orthogonal polynomials \cite{Mourad}.
\end{proof}

To prove Proposition \ref{Rexpansions} is equivalent to prove that $\wt R(z)\sim R^{[p]}(z)$ for $p=\infty,0$ where
\be
\label{eq:TRT}
\wt R(z)=TR(z)T^{-1}.
\ee
It follows from the previous proposition that $\wt R(z)$ satisfies
\be
\label{LaxTR}
\frac{\pa}{\pa z}\wt R(z)=[\wt U(z),\wt R(z)]
\ee
where $\wt U(z)=TU(z)T^{-1}=\wt U_0/z+\wt U_1/(1-z)$, with
\begin{align}
\nonumber
\wt U_0&=T U_0 T^{-1} = \frac{1}{\a+\b+2N}
\begin{pmatrix}
 \frac{2 N (\a+\b+N)+\a (\a+\b)}{2} & -N(\a+N)(\b+N)(\a+\b+N) \\
 1 & -  \frac{2 N (\a+\b+N)+\a (\a+\b)}{2}
\end{pmatrix},
\\
\label{TATmatrix}
\wt U_1&= T U_1 T^{-1} = \frac{1}{\a+\b+2N}
\begin{pmatrix}
- \frac{2 N (\a+\b+N)+\b (\a+\b)}{2} & -N(\a+N)(\b+N)(\a+\b+N) \\
 1 &  \frac{2 N (\a+\b+N)+\b (\a+\b)}{2}
\end{pmatrix}.
\end{align}

Introduce the matrices
\be
\label{sl2basis}
\s_3=\left(
\begin{array}{cc}
1 & 0 \\ 0 &-1
\end{array}
\right),\qquad
\s_+=\left(
\begin{array}{cc}
0 & 1 \\ 0 &0
\end{array}
\right),\qquad
\s_-=\left(
\begin{array}{cc}
0 & 0 \\ 1 &0
\end{array}
\right),
\ee 
and write
\begin{equation}
\label{R_decom}
\wt R(z)=\frac 12 \1+r_3\s_3+r_+\s_++r_-\s_-,\qquad \wt U(z)=u_3\s_3+u_+\s_++u_-\s_-,
\end{equation}
where we used that $\tr R(z)= 1, \tr U(z)=0$. For the sake of brevity we omit the dependence on $z$ in the $\mathfrak {sl}_2$ components.
The Lax equation \eqref{LaxTR} yields the coupled first order linear ODEs
\be \label{syseq}
\pa_z r_3=u_+r_--u_-r_+,\qquad
\pa_z r_+=2(u_3r_+-u_+r_3),\qquad
\pa_z r_-=2(u_-r_3-u_3r_-),
\ee
which is equivalent to three decoupled third order linear ODEs, one for $\pa_zr_3$ \small
\begin{multline}
\label{Rec3}
3 \left[2N(\a+\b+N)+\a(\a+\b)-2 -z \left((\a+\b+2 N)^2-4\right)\right]\pa_z r_3\\  - \left[\a^2-4-2 z\left(2 N (\a+\b+N)+\a (\a+\b)-12\right)+z^2 \left((\a+\b+2 N)^2-24\right)\right]\pa_z^2 r_3\\
- 5 z(z-1) (1-2 z)  \pa_z^3 r_3  +z^2(z-1)^2 \pa_z^4 r_3 =0,
\end{multline} \normalsize
and for $r_\pm$ \small
\begin{multline}
\label{RecPM}
\left[2 N (\a+\b+N\pm1)+(\a\pm1) (\a+\b)-z (\a+\b+2 N\pm2) (\a+\b+2 N)\right]r_\pm\\
- \left[\a^2-1 -z \left(-4 N (\a+\b+N\pm1)-2(\a+\b)(\a\pm1)+6 \right) \right. \\ \left.+ z^2 \left(4 N (\a+\b+N\pm1)+(\a+\b)(\a+\b\pm2)-6\right)  \right]\pa_z r_\pm +3z(z-1)(2z-1)\pa_z^2 r_\pm+z^2 (z-1)^2\pa_z^3 r_\pm =0.
\end{multline} \normalsize
The following ansatz is quite natural in view of our previous work \cite{GGR2020} about the Laguerre Unitary Ensemble (see also \cite{DY2017} for the Gaussian Unitary Ensemble);
namely we write the expansions of the entries of $R(z)$ at $z=\infty$ as
\begin{align}
r_3(z)&\sim\frac 12+\frac{1}{\a+\b+2N}\sum_{\ell\geq 0} \frac{1}{z^{\ell+1}} \ell A_\ell(N)
,		\\
r_+(z)&\sim \frac{1}{\a+\b+2N}\sum_{\ell\geq 0} \frac{N(\a+N)(\b+N)(\a+\b+N)}{z^{\ell+1}} B_\ell(N+1),		\\
r_-(z)&\sim -\frac{1}{\a+\b+2N}\sum_{\ell\geq 0} \frac{1}{z^{\ell+1}} B_\ell(N), \label{ansatz}
\end{align}
for some coefficients $A_\ell(N)=A_\ell(N,\a,\b)$ and $B_\ell(N)=B_\ell(N,\a,\b)$.
By substitution in \eqref{Rec3} and \eqref{RecPM} we see that the ansatz is consistent with them; in particular we get the following \emph{three term recurrence relations} for $A_\ell(N),B_\ell(N)$;
\begin{align} \nonumber
&(2\ell+1)\left(\a(\a+\b)-\ell(\ell+1)+2N(\a+\b+N)\right)A_{\ell}(N)\\ \label{recAl}
&\qquad+(\ell-1)(\ell^2-\a^2)A_{\ell-1}(N)+(\ell+2)\left((\ell+1)^2-(\a+\b+2N)\right)A_{\ell+1}(N)=0,
\\
\nonumber
&(2\ell+1)\left((\a- 1)(\a+\b)-\ell(\ell+1)+2N(\a+\b+N-1)\right)B_{\ell}(N)\\  \label{recBl}
&\qquad+\ell(\ell^2-\a^2)B_{\ell-1}(N)+(\ell+1)\left((\ell+1)^2-(\a+\b+2N-1)\right)B_{\ell+1}(N)=0,
\end{align}
for $\ell\geq 1$, together with the initial conditions \small
\begin{align} 
A_0(N,\a,\b)&=\frac{N(\b+N)}{\a+\b+2N},	&	 A_1(N,\a,\b)=\frac{N(\a+N)(\b+N)(\a+\b+N)}{(\a+\b+2 N-1) (\a+\b+2 N) (\a+\b+2 N+1)},	\\
B_0(N,\a,\b)&= \frac{1}{(\a+\b+2N-1)},	&	 B_1(N,\a,\b)= \frac{(\a-1) (\a+\b)+2 N (\a+\b+N-1)}{(\a+\b+2 N-2) (\a+\b+2 N-1) (\a+\b+2 N)}.
\end{align} \normalsize
The initial conditions are obtained from \eqref{monicJacobi} and \eqref{CauchyInf}.
It can be checked that the recurrence relation for the coefficients of $r_+(z)$ are actually those of $r_-(z)$, modulo a shift in $N$, as claimed in \eqref{ansatz}.

The three term recurrence relations \eqref{recAl} and \eqref{recBl} can be solved in terms of \emph{Wilson polynomials} \eqref{HyperW}.

\begin{proposition} \label{Wilson}
The coefficients $A_\ell(N,\a,\b)$ and $B_\ell(N,\a,\b)$ can be expressed in terms of  \emph{Wilson Polynomials}, defined in \eqref{HyperW}, as
\small
\begin{align*}
&A_\ell(N,\a,\b)=\frac{(-1)^{N-1}(\a+\ell)!(\a+\b+N)!(\b+N)}{(N-1)!(\a+N-1)!(\a+\b+2N+\ell)!}\, W_{N-1}\biggl(-\left(\ell+ \frac 12 \right)^2;\frac 32, \frac 12,\a+\frac12,\frac 12-\a-\b-2N\biggr), 	\\
&B_\ell(N,\a,\b)=\frac{(-1)^{N-1}(\a+\ell)!(\a+\b+N-1)!}{(N-1)!(\a+N-1)!(\a+\b+2N+\ell-1)!}\, W_{N-1}\biggl(-\left(\ell+ \frac 12 \right)^2;\frac 12, \frac 12,\a+\frac12,\frac 32-\a-\b-2N\biggr).
\end{align*}\normalsize
This is equivalent to the hypergeometric representation
\small
\begin{align*}
A_\ell(N,\a,\b)&=N (\a+N) (\b+N) (\a+\b+N) \frac{(\a+2)_{\ell-1}}{(\a+\b+2 N-1)_{\ell+2}} \pFq{4}{3}{1-\ell,\ell+2,1-\b-N,1-N}{2,\a+2,2-\a-\b-2 N}{1}, \\
B_\ell(N,\a,\b)&= \frac{(\a+1)_\ell}{(\a+\b+2 N-1)_{\ell+1}} \pFq{4}{3}{-\ell,\ell+1,1-\b-N,1-N}{1,\a+1,2-\a-\b-2 N}{1}. 
\end{align*}\normalsize
\end{proposition}

\begin{proof}
The identification with the Wilson polynomials is obtained by comparing the recurrence relations \eqref{recAl} and \eqref{recBl} with the difference equation for this family of orthogonal polynomials, which reads
\be
n(n+a+b+c+d-1)w(k)=C(k)w(k+\i)-\left[C(k)+D(k)\right]w(k)+D(k)w(k-\i),
\ee
where $w(k)=W_n(k^2;a,b,c,d)$ and
\be
C(k)=\frac{(a-\i k)(b-\i k)(c-\i k)(d-\i k)}{2\i k (2 \i k -1)},\quad
D(k)=\frac{(a+\i k)(b+\i k)(c+\i k)(d+\i k)}{2\i k (2 \i k +1)}.
\ee
The hypergeometric representation of $A_\ell,B_\ell$ then directly follows from that of the Wilson polynomials in \eqref{HyperW}.
\end{proof}

The above Proposition, together with the expansions \eqref{ansatz}, yields the first part of Proposition \ref{Rexpansions}. The asymptotics of $R(z)$ at $z=0$ are obtained in a similar way. 
More precisely, we claim that the expansion at $z=0$ of the entries of $\wh R(z)$ reads as
\begin{align}
\nonumber
r_3(z)&\sim\frac 12+\frac{1}{\a+\b+2N} \sum_{\ell\geq 0} \frac{(\a+\b+2N-\ell)_{2\ell+1}}{(\a-\ell)_{2\ell+1}} (\ell+1) A_\ell(N,\a,\b) z^\ell
,		\\ 
\nonumber
r_+(z)&\sim -\frac{1}{\a+\b+2N} \sum_{\ell\geq 0} N(\b+N)(\a+N)(\a+\b+N) \frac{(\a+\b+2N+1-\ell)_{2\ell+1}}{(\a-\ell)_{2\ell+1}} B_\ell(N+1,\a,\b) z^\ell,		\\
r_-(z)&\sim \frac{1}{\a+\b+2N} \sum_{\ell\geq 0}  \frac{(\a+\b+2N-1-\ell)_{2\ell+1}}{(\a-\ell)_{2\ell+1}} B_\ell(N,\a,\b) z^\ell.
\label{ansatz0}
\end{align}
This can be proven by checking that plugging  the formul\ae\ \eqref{ansatz0}  in the equations \eqref{Rec3}, \eqref{RecPM}, one obtains   the same recurrence relations \eqref{recAl} and \eqref{recBl}.
The associated initial conditions can again be computed from \eqref{monicJacobi} and \eqref{Cauchy0}. This concludes the proof of Proposition \ref{Rexpansions}.\hfill$\square$

\subsection{Proof of Corollary \ref{corollaryWilson}} \label{sec42}

\subsubsection{Case \texorpdfstring{$\ell=1$}{l=1}}

From Theorem \ref{thmcumulants} we write the formula for $\mathscr C_1(z)$ by using the differential equation \eqref{Ydiffeq} as
\be\label{111}
\mathscr C_1(z) =\tr\left(Y^{-1}_N(z)Y_N'(z)\mathrm{E}_{1,1}\right)=\tr(U(z)R(z))-\frac 12\left(\frac \a z-\frac \b{1-z}\right),
\ee
where $E_{1,1}:=\begin{pmatrix} 1 & 0 \\ 0 & 0 \end{pmatrix}$. In the last step we used the definition of $R(z)=Y_N(z)E_{1,1}Y_N(z)^{-1}$, the cyclic property of the trace and the equation
\be
Y'_N(z)=U(z)Y_N(z)-\left(\frac \a z-\frac \b{1-z}\right)Y_N(z)\frac{\s_3}2,
\ee
which follows from \eqref{Ydiffeq}.

\begin{lemma}
We have
\be \label{OnePointEq}
\pa_z\left[z(1-z) \tr(U(z)R(z))\right] = -(\a+\b+2N) \tr\left(R(z)\frac{\s_3}{2}\right).
\ee 
\end{lemma}
\begin{proof}
We compute
\be
\pa_z\left[z(1-z) \tr(U(z)R(z))\right]=(1-2z)\tr(U(z)R(z))+z(1-z)\tr(U'(z)R(z))+z(1-z)\tr(U(z)R'(z)).
\ee
The last term vanishes due to the Lax equation \eqref{LaxR}, because $\tr(U(z)[U(z),R(z)])=0$ by the cyclic property of the trace. Then we use the identity
\be
(1-2z)U(z)+z(1-z)U'(z)= - \frac{\a+\b+2N}{2}\s_3,
\ee
which can be checked directly from \eqref{eq:U}. The proof is complete.
\end{proof}

By this lemma and \eqref{111} we obtain
\be
\pa_z[z(1-z)\mathscr C_1(z)]=-(\a+\b+2N) \left(R_{1,1}(z)-\frac 12\right)+\frac {\a+\b}2=-(\a+\b+2N) \left(R_{1,1}(z)-1\right)-N,
\ee
where we use that $\tr R(z)=1$ to compute $\tr(R(z)\s_3)=2R_{1,1}(z)-1$, and we denote $R_{1,1}$ the $(1,1)$-entry of $R$.
Integrating this identity implies that for any $p\in\C\setminus[0,1]$ we have
\be
\label{zp}
z(1-z)\mathscr C_1(z)-p(1-p)\mathscr C_1(p)=(\a+\b+2N)\int_p^z\left(1-R_{1,1}(w)\right)\d w+N(p-z).
\ee

Letting $p\to 0$ in \eqref{zp} we have $p(1-p)\mathscr C_1(p)\to 0$ and so
\be
\mathscr C_1(z)=\frac{(\a+\b+2N)}{z(1-z)}\int_0^z\left(1-R_{1,1}(w)\right)\d w-\frac{N}{1-z}.
\ee
Expanding this identity at $z=0$ we get at the left hand side
\be
\mathscr C_1(z)\sim -\sum_{k\geq 0}\left\langle\tr X^{-k-1}\right\rangle z^{k}=\mathscr F_{1,0}(z),
\ee
and using Proposition \ref{Rexpansions} (note that $(TR(z)T)_{1,1}=R_{1,1}(z)$ because $T$ is diagonal) the formula for $\mathscr F_{1,0}(z)$ is proved.

Letting instead $p\to\infty$ we have $p(1-p)\mathscr C_1(p)\sim (1-p)N-\left\langle\tr X\right\rangle+\O(1/p)$ and therefore from \eqref{zp} we have (noting that $R_{1,1}(w)=1+\O(w^{-2})$ so the integral is well defined)
\be
z(1-z)\mathscr C_1(z)=(\a+\b+2N)\int_\infty^z\left(1-R_{1,1}(w)\right)\d w+(1-z)N-\left\langle\tr X\right\rangle.
\ee
We can compute
\be
\left\langle\tr X\right\rangle=\frac{N(\a+N)}{\a+\b+2N}
\ee
by expanding the general formula $\mathscr C_1(z)=\tr\left(Y_N^{-1}(z)Y_N'(z)\s_3/2\right)$ at $z=\infty$, using \eqref{monicJacobi} and the first few terms in \eqref{CauchyInf}. We finally obtain
\be
\mathscr C_1(z)=\frac{\a+\b+2N}{z(1-z)}\int_\infty^z\left(1-R_{1,1}(w)\right)\d w+\frac Nz -\frac{N(\a+N)}{z(1-z)(\a+\b+2N)}.
\ee
Expanding this identity at $z=\infty$ we get at the left hand side
\be
\mathscr C_1(z)\sim \sum_{k\geq 0}\frac{\left\langle\tr X^{k}\right\rangle}{z^{k+1}}=\frac Nz+\mathscr F_{1,\infty}(z)
\ee
and using Proposition \ref{Rexpansions} (again note that $(TR(z)T)_{1,1}=R_{1,1}(z)$ because $T$ is diagonal) the formula for $\mathscr F_{1,\infty}(z)$ is also proved.

\subsubsection{Case \texorpdfstring{$\ell\geq 2$}{l>1}}

In this case we note that
\be
\mathscr C_\ell^{\sf c}(z_1,\dots,z_\ell)=-\sum_{(i_1,\dots,i_\ell)\in\cyc((\ell))}\frac{\tr\left(\wt R(z_{i_1})\cdots \wt R(z_{i_\ell})\right)}{(z_{i_1}-z_{i_2})\cdots(z_{i_\ell}-z_{i_1})}-\frac{\delta_{\ell,2}}{(z_1-z_2)^2},\qquad \ell\geq 2
\ee
where $\wt R(z)=TR(z)T^{-1}$ as in \eqref{eq:TRT}. We now expand both sides of this identity at $z=0,\infty$. The expansion of the right hand side follows from Proposition \ref{Rexpansions} which asserts that $\wt R(z)\sim R^{[0]}(z),R^{[\infty]}(z)$ as $z\to 0,\infty$, respectively. For the left hand side instead, at $z\to 0$ we have
\be
\mathscr C_\ell^{\sf c}(z_1,\dots,z_\ell)\sim(-1)^\ell\sum_{k_1,\dots,k_\ell\geq 0}\left\langle\tr X^{-k_1-1}\cdots\tr X^{-k_\ell-1}\right\rangle^{\sf c}z_1^{k_1}\cdots z_\ell^{k_\ell}=\mathscr F_{\ell,0}^{\sf c}(z_1,\dots,z_\ell),
\ee
while at $z\to\infty$ we have
\be
\mathscr C_\ell^{\sf c}(z_1,\dots,z_\ell)\sim \sum_{k_1,\dots,k_\ell\geq 0}\left\langle\tr X^{k_1}\cdots\tr X^{k_\ell}\right\rangle^{\sf c}z_1^{-k_1-1}\cdots z_\ell^{-k_\ell-1}=\mathscr F_{\ell,\infty}^{\sf c}(z_1,\dots,z_\ell),
\ee
where in the last identity we use that terms with $k_i=0$ for some $i$ do not contribute to the sum; indeed the connected correlator $\left\langle\tr X^{k_1}\cdots\tr X^{k_\ell}\right\rangle^{\sf c}$ vanishes whenever $k_i=0$ for some $i$. The proof is complete.\hfill$\square$

\subsection*{Acknowledgements}
We thank M. Bertola and D. Yang for valuable conversations.
This project has received funding from the European Union's H2020 research and innovation programme under the Marie Sk\l odowska--Curie grant No. 778010 {\em  IPaDEGAN}. The research of G.R. is supported by the Fonds de la Recherche Scientifique-FNRS under EOS project O013018F.

\end{document}